\documentclass{article}
\usepackage[a4paper, total={6in, 8in}]{geometry}
\usepackage[dvipsnames]{xcolor}
\usepackage{graphicx}
\usepackage{amsmath,amssymb,amsfonts}
\usepackage{todonotes}
\usepackage{hyperref}
\usepackage{caption}
\usepackage{subcaption}
\usepackage{verbatim}
\usepackage{algorithm}
\usepackage{algpseudocode}
\usepackage{csquotes}
\usepackage{bbm}
\usepackage[british]{babel}

\usepackage{amsthm}

\newtheorem{proposition}{Proposition}
\newtheorem{lemma}{Lemma}

\usepackage[style=ieee, url=false, maxbibnames=20]{biblatex}
\DeclareFieldFormat[article,misc,online,book,report,incollection]{title}{\mkbibemph{#1}}
\AtEveryBibitem{\clearlist{language}}
\AtEveryBibitem{\clearfield{month}}
\AtEveryBibitem{\clearfield{issn}}

\title{Travelling waves modulated by subthreshold oscillations in networks of integrate-and-fire neurons}
\author{Henry D. J. Kerr$\,^{1,3}$, Peter Ashwin$\,^{2,3}$, Kyle C.~A.~Wedgwood$\,^{1,2,3}$}
\date{\small
$\,^1$ Living Systems Institute, University of Exeter, Exeter, EX4 4QD, United Kingdom\\
$\,^2$ EPSRC Hub for Quantitative Modelling and Healthcare, University of Exeter, Exeter, EX4 4QD, United Kingdom\\
$\,^3$ Department of Mathematics and Statistics, Harrison Building, University of Exeter, North Park Road, Exeter, EX4 4QF, United Kingdom}

\begin{document}

\maketitle

\begin{abstract}
\noindent Travelling waves of neural firing activity are observed in brain tissue as a part of various sensory, motor and cognitive processes.  They represent an object of major interest in the study of excitable networks, with analysis conducted in both neural field models and spiking neuronal networks.  The latter class exposes the single-neuron dynamics directly, allowing us to study the details of their influence upon network-scale behaviour.  
Here we present a study of a laterally-inhibited network of leaky integrate-and-fire neurons modulated by a slow voltage-gated ion channel that acts as a linear adaptation variable.  As the strength of the ion channel increases, we find that its interaction with the lateral inhibition increases wave speeds.  The ion channel can enable subthreshold oscillations, with the intervals between the firing events of loosely-coupled travelling wave solutions structured around the neuron's natural period.  These subthreshold oscillations also enable the occurrence of codimension-2 grazing bifurcations; along with the emergence of fold bifurcations along wave solution branches, the slow ion channel introduces a variety of intermediate structures in the solution space.  
These point towards further investigation of the role neighbouring solution branches play in the behaviour of waves forced across bifurcations, which we illustrate with the aid of simulations using a novel root-finding algorithm designed to handle uncertainty over the existence of firing solutions.  
\end{abstract}

\section{Introduction}
Understanding collective behaviour in excitable networks, such as networks of neurons or cardiac cells, is a central question in network dynamics.
In neuroscience, collective behaviour in the form of travelling waves is believed to underpin computation in the brain across a variety of scales, with waves being observed in a range of sensory, motor and cognitive systems~\cite{Muller2018}.
These travelling waves have attracted much interest from the modelling community, with seminal research by Amari~\cite{Amari1977}, Ermentrout~\cite{ermentroutNeuralNetworksSpatiotemporal1998}, Bressloff~\cite{bressloffWavesNeuralMedia2014}, and Coombes and others~\cite{NeuralFields} providing techniques for establishing the existence and stability of travelling waves in networks of neural populations coupled in a nonlocal fashion through long-range synapses.
These results have paved the way for understanding how different facets of synaptic communication, such as transmission delays, recurrent inhibition, and anisotropy, impact the propagation of neural activity across cortical tissue.

Many studies of travelling waves in neuronal networks have employed an activity- or rate-based approach such as that used by \textcite{Amari1977}, in which the models abstractly describe the average firing rate of the local population of neurons at each point in the network.  This simplifies the local dynamics of the model, greatly facilitating analysis and in some cases admitting closed-form solutions for travelling waves.  However, such simplification leaves out the processing that occurs in small circuits of neurons or even within single neurons, such as resonant responses that suppress or promote inputs with different frequencies~\cite{Stark2022}.  At the other end of the complexity scale, models that follow the Hodgkin--Huxley formalism \cite{Hodgkin1952e} explicitly describe the various ion channel dynamics that generate action potentials (spikes), sacrificing analytical tractability in favour of providing a detailed model of the underlying biological processes.

Between these two extremes we find the integrate-and-fire (IF) models, which explicitly describe firing events using a hard reset condition but forgo the prescription of specific ion channel kinetics in favour of simpler and more abstract representations.
At the single cell level, IF models can capture complex firing patterns such as bursting dynamics~\cite{Izhikevich2001} and can be fit to experimental datasets with high precision~\cite{Badel2008a}.
IF models thereby present a useful bridge between rate-based models and more complicated spiking models.
Travelling wave solutions are also found in networks of IF models incorporating the same nonlocal synaptic coupling as in the Amari-style neural field models.

The simplicity of IF models often facilitates the construction of semi-explicit solutions that are written in terms of the times at which neurons in the population fire.
This distinction enables exploration of the contribution of individual firing events to overall collective behaviour.
One key example revolves around so-called ``bump attractors'', spatially localised solutions of heightened firing activity that are thought to underpin working memory~\cite{Wimmer2014}.
In Amari-style networks, bump attractors appear as time-invariant profiles, whereas in IF networks these are represented by solutions undergoing spatiotemporal chaos or \textit{chimeras}~\cite{Omelchenko2024}.
Importantly, this chaotic nature is not derived from a finite-size effect and instead occurs due to the network trajectories in the solution space moving between different travelling wave solutions, each of which is of saddle-type~\cite{Avitabile2023}.
This rich dynamical behaviour further highlights the importance of understanding travelling wave dynamics for probing neural computation.

Subthreshold oscillations and resonances in local dynamics can have a significant impact in travelling wave behaviour, such as through direct correlations between resonant frequencies and wave propagation speed~\cite{shanTheoreticalExperimentalStudy2025}.
The interplay between resonance and travelling wave behaviour has been hypothesised to play a key role in auditory processing in the cochlea~\cite{bellHearingTravellingWave2004} with later studies providing theoretical~\cite{bellResonanceApproachCochlear2012} and experimental~\cite{nankaliInterplayTravelingWave2022} support for this.
Correlation between resonances and wave propagation is also observed in other excitable tissues, such as cardiac tissue, in which local frequency tuning determines whether waves are initiated and terminated~\cite{tepleninAtypicalCollectiveOscillatory2025}.
At larger scales, nonlinear resonances in coupling have been shown to promote the emergence of complex mixed-mode propagating oscillations in large scale brain models~\cite{galinskyUniversalTheoryBrain2020}.

In this article, we examine how the incorporation of subthreshold oscillations in a network of integrate-and-fire neurons coupled synaptically with lateral inhibition impacts the speeds and bifurcation structures of travelling waves supported by the network.
Incorporation of this subthreshold oscillation is achieved by adding a second state variable with linear dynamics, representing the action of a voltage-gated ion channel such as an HCN \cite{Nolan2007} or Kv1 channel \cite{Sciamanna2011}.  This variable is akin to an adaptation variable in neural field models~\cite{ermentroutSpatiotemporalPatternFormation2014}.
This choice was inspired by studies in isolated neuron models that quantified the mechanisms generating resonances at local scales~\cite{Richardson2003, Rotstein2014}, and multi-scale studies that demonstrated that resonances at the network level can be generated across multiple neural scales independently of one another~\cite{Stark2022}.
Restricting the new dynamics to be linear allows us build upon our previous results on networks with one-dimensional local dynamics~\cite{Avitabile2023} to provide semi-explicit equations for the construction and stability of travelling waves.
We leverage these results to explore wave solutions as model parameters are varied, checking our results against numerical simulations using a novel event-based algorithm that makes use of root-finding rather than numerical integration to significantly reduce computational cost and ensure that all firing events are captured accurately.

The remainder of the manuscript is organised as follows.
In Sect.~\ref{sec:model}, we describe our model and review the relevant properties of the local model.
In Sect.~\ref{sec:numerics}, we briefly describe our approach for efficient and accurate simulation of our network model, with full algorithmic details provided in the Appendix.
In Sect.~\ref{sec:tw_construction}, we construct travelling waves and their stability.
In Sect.~\ref{sec:results}, we investigate how wave solution structures, speeds, stabilities and bifurcations vary as model parameters are changed through computation of bifurcation diagrams and numerical simulation.
We end in Sect.~\ref{sec:discussion} with concluding remarks and suggestions for future work.

\section{The neuronal network model}
\label{sec:model}

\noindent 
Taking inspiration from~\cite{Laing2001,Avitabile2023}, our model comprises $N$  homogeneous neurons equispaced over a ring domain $\mathbb{T} = \mathbb{R} / 2L\mathbb{Z}$ with spatial period $2L$.
The individual neurons are positioned at $x_n = -L + \Delta_xn$, for $n=1,2,\dots,N$, where $\Delta_x = 2L/N$.
The local dynamics for each neuron is prescribed by a variant of the leaky integrate-and-fire model presented in~\cite{Richardson2003, Rotstein2014} that, in addition to the voltage variable $v$, incorporates a variable $u$ representing the linearised activity of a voltage-gated ion channel.
The inclusion of a second state variable facilitates inclusion of subthreshold oscillatory dynamics, be it sustained or decaying. 
Such subthreshold dynamics has recently been shown to facilitate phase-locking to external stimuli, which may in turn be relevant to the oscillatory interference models used to describe spatial navigation~\cite{makarenkovBifurcationSpikingOscillations2025}.
The local dynamics is completed through the addition of a synaptic buffer variable $s$, resulting a system of three ODEs,
\begin{align}
    \frac{dv_n}{dt} &= I - v_n - u_n + s_n - (v_\text{th} - v_\text{r}) \sum_{k \in \mathbb{Z}} \delta (t - t_{n,k} ), \label{eq:v_discrete} \\
    \frac{du_n}{dt} &= Rv_n - Du_n, \label{eq:u_discrete} \\
    \frac{ds_n}{dt} &= -\beta s_n + \beta f^\text{in}_n(t), \label{eq:s_discrete}
\end{align}
for $n = 1, 2, \dots N$, and where $\delta$ is the Dirac delta function.
When $v_n$ reaches the threshold value $v_\text{th}$, the neuron `fires' and $v_n$ is instantaneously set to the reset value $v_\text{r} < v_\text{th}$. Throughout this work we shall take $v_\text{th} = 1$, $v_\text{r} = 0$.
We denote the time at which the $n$th neuron fires for the $k$th time by
\begin{equation}
    t_{n,k} := \inf\{ t \in \mathbb{R} \mid v_n(t) = v_\text{th}, \, t > t_{n,k-1}\}.
    \label{eq:firing_time}
\end{equation}
The secondary variable $u_n$ models the current flow through a voltage-gated ion channel, which may be linearised from conductance-based models \cite{Richardson2003}.  Increases in $v_n$ drive increases in $u_n$, which subsequently depresses $v_n$, modelling either a hyperpolarisation-activated inward flow of charge such as an HCN channel \cite{Nolan2007} or a depolarisation-activated outward flow of charge such as a Kv1 channel \cite{Sciamanna2011}.  We parametrise the dynamics of $u_n$ by a response rate $R \geq 0$ to changes in $v_n$ and a decay rate $D > 0$.  This varies from the more common conductance parametrisation of the linear channel \cite{Richardson2003, Brunel2003, Rotstein2014} to better enable visual display of solutions, retain comparability with prior studies of travelling waves and bumps in LIF neurons \cite{Avitabile2023, Laing2001}, and relate the parameter spaces of natural oscillation and resonant response (see Figure \ref{fig:stability}).

The synaptic variable $s_n$ models the response of neuron $n$ to synaptic inputs from the rest of the neuronal population, with the parameter $\beta > 0$ controlling its timescale.  More precisely, \eqref{eq:s_discrete} describes a normalised first-order temporal filter for incoming inputs (represented by $f_n^\text{in}$), wherein $s_n$ responds to inputs by first jumping instantaneously by an amount proportional to $\beta$ before exponentially decaying to rest with rate parameter $\beta$.
This filter gives a physiologically agreeable form to the response of $v$ to synaptic inputs while allowing $s_n$ to act as an input buffer in place of tracking the timecourse of individual inputs during simulation.
Other models for temporal filters, including those with higher order responses can be found in \cite{Ermentrout1998}.

The parameter $I$ represents an applied current and sets the excitability of all neurons.  In the absence of synaptic inputs, i.e., with $f^{\text{in}}_n(t)=0$, the neuron settles to the steady state 
\begin{equation}
    v_\text{rest} = \frac{DI}{R + D}, \quad u_\text{rest} = \frac{RI}{R + D}, \quad s_\text{rest} = 0.
\end{equation} 
We will often wish to treat $v_\text{rest}$ as an independent parameter to be declared while varying $R$ and $D$, in which case we can redefine $I := ((R + D)/D)v_\text{rest}$.

\subsection{Synaptic connectivity}
Connections between neurons in the network are modelled using a distance-dependent all-to-all coupling rule.
This is achieved by setting the input function $f^\text{in}_n$ to
\begin{equation}
    f^\text{in}_n(t) = \Delta_x \sum_{k \in \mathbb{Z}} \sum_{m \neq n}  w\big(|{x}_n - {x}_m|\big) \delta \big(t - t_{m,k}\big),
\label{eq:discrete-input}
\end{equation}
recalling that $\Delta_x = 2L/N$ is the reciprocal of the neuron density, i.e.\ the free space around each neuron.
The distance-based connectivity kernel $w : \mathbb{T} \to \mathbb{R}$ is defined as a difference of Gaussians,
\begin{equation}
    w(d) = \frac{A}{ a\sqrt{2\pi} } e^{-\frac{d^2}{2a^2}}
    - \frac{B}{ b\sqrt{2\pi} } e^{-\frac{d^2}{2b^2}}.
\label{eq:mexican-hat}
\end{equation}
For $L \gg b > a > 0$ and $A, B > 0$, \eqref{eq:mexican-hat} describes a Mexican hat function in which nearby neurons excite one another (positive connection weights), neurons separated by an intermediate distance inhibit one another (negative connection weights), and neurons at longer distances have negligible interaction (weights approach zero supraexponentially as $d \to \pm \infty$).  This models central excitation with lateral or surround inhibition, a pattern of functional connectivity that is observed in sensory cortices \cite{DelRosario2025, Studer2022, Shmuel2019} and can be retrieved as a general simplification of functional connectivity in models of excitatory and inhibitory subpopulations wherein the inhibitory-to-inhibitory input is negligible and the excitatory-inhibitory-excitatory relay has a longer average length than the direct excitatory-to-excitatory connection \cite{Kang2003, Wilson1973}.  We take $A = B$ to balance excitation and inhibition across the network such that $\int_\mathbb{R}w(x)dx = 0$.  For $L \gg b$ the connectivity range is shorter than the length of the ring, meaning the effects of the leading and trailing inhibition do not overlap and we may approximate the ring domain $\mathbb{T}$ with $\mathbb{R}$.

\begin{figure}
    \centering
    \includegraphics[width=0.5\textwidth]{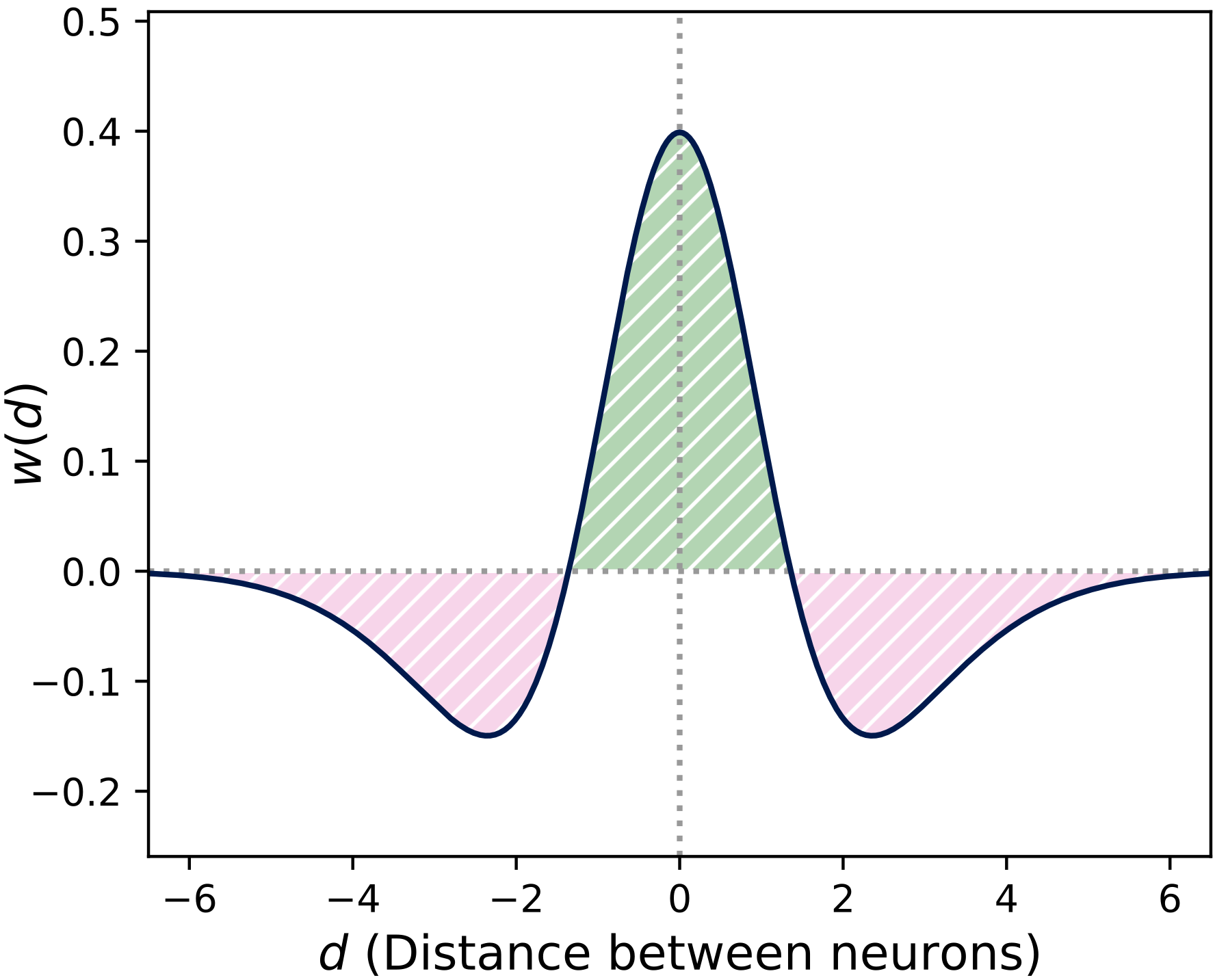}
    \caption{Mexican Hat function \eqref{eq:mexican-hat} for our selected parameters $A = B = 2$, $a = 1$, $b = 2$.}
    \label{fig:mexican-hat}
\end{figure}

\subsection{Model solutions}

The linearity of the single-neuron model \eqref{eq:v_discrete}--\eqref{eq:s_discrete}
permits the construction of explicit solutions. Introducing the vector $\mathbf{v}_n = (v_n, u_n)^\top$, we can rewrite the system of equations \eqref{eq:v_discrete}--\eqref{eq:s_discrete} in the form
\begin{equation}
    \dot{\mathbf{v}}_n = \mathbf{M} \mathbf{v}_n + \mathbf{b}_n,
    \label{eq:vu_matvec}
\end{equation}
where
\begin{equation}
    \mathbf{M} = 
    \begin{pmatrix}  -1 & -1 \\ R & -D \end{pmatrix}, \quad \mathbf{b}_n = \begin{pmatrix}
        I + s_n - (v_\text{th} - v_\text{r}) \sum_{n=1}^\infty \delta (t - t_{k,n})\\
        0
    \end{pmatrix}.
\end{equation}
Here we see that $s_n(t)$ acts as a non-autonomous forcing variable for the $(v_n,u_n)$ subsystem.

System \eqref{eq:vu_matvec} can be solved via variation of constants to give
\begin{equation}
    \textbf{v}_n = e^{t\textbf{M}}  \Bigg( \mathbf{v}_n(0) + 
    \int_0^t e^{-t'\textbf{M}} \begin{pmatrix} 1 \\ 0 \end{pmatrix} 
    \Big( I + s_n(t') - (v_\text{th} - v_\text{r}) \sum_{n=1}^\infty \delta (t' - t_n) \Big) dt'
     \Bigg),
\end{equation}
where $e^{t \mathbf{M}}$ is the matrix exponential of $t \mathbf{M}$.
Matrix $\mathbf{M}$ has eigenvalues
$\lambda_{1} = -p - q$ and $\lambda_2 = -p + q$, where
\begin{equation}
    p = \frac{1}{2}(D + 1), \quad
    q = \frac{1}{2}\sqrt{(D - 1)^2 - 4R}.
    \label{eq:eigenvalues}
\end{equation}

\begin{figure}
    \centering
    \includegraphics[width=\linewidth]{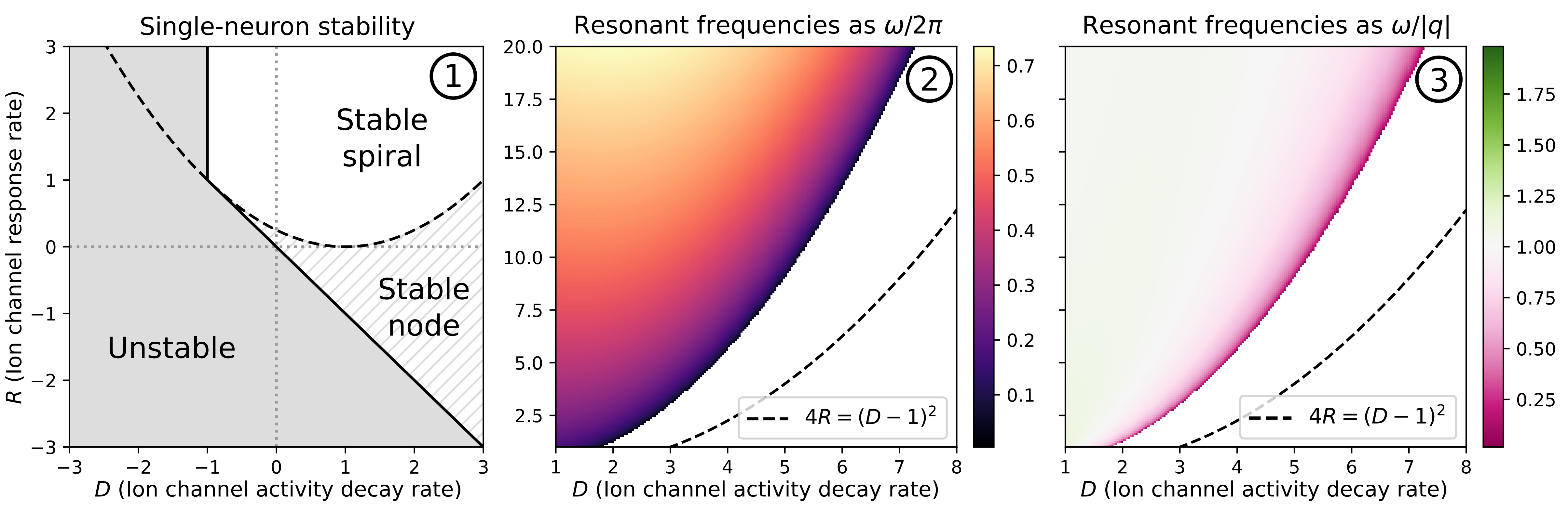}
    \caption{1: Stability diagram for the fixed point of the single-neuron model in $(R, D)$ parameter space.  The dashed parabola represents the relationship $4R = (D - 1)^2$, splitting the parameter space into regions in which solutions to \eqref{eq:vu_matvec} are trigonometric (above) or hyperbolic (below).  2: Resonant frequencies of neurons for values of $R$ and $D$; white regions do not show a resonant response.  3: Resonant frequencies as a proportion of the neuron's natural frequency $|q|/2\pi$.}
    \label{fig:stability}
\end{figure}

In the absence of forcing through $s(t)$, the system has a fixed point, for which the stability is shown in Figure \ref{fig:stability}.1.
The fixed point is stable for $D > -1$ and $R + D > 0$, and is a stable spiral if $4R > (D - 1)^2$, and a stable node otherwise.
We wish to focus on the case in which solutions have a decaying oscillatory component, so we shall restrict ourself to the case $4R > (D - 1)^2$.  We note that despite our reparametrisation the stability diagram is identical to that of prior work \cite{Richardson2003, Rotstein2014}, as at rest our voltage and current equations have the same form as their current and voltage equations respectively.
  
Increasing $R$ within the oscillatory case only increases the imaginary portion of our eigenvalues $q$, increasing the natural frequency of our neuron without increasing the rate of exponential convergence towards the steady state.  As the oscillatory regime has a quadratic boundary, increasing the ion channel timescale by increasing $R$ and $D$ while holding $R/D$ constant will eventually abolish natural oscillation.

\subsection{Solutions between firing events}
Between firing events the coupling term is $f_n^\text{in} = 0$, so we may treat the neurons as being isolated.
Without loss of generality, we can take our initial condition as being at time $t = 0$, and integrate \eqref{eq:s_discrete} to find $s_n(t) = s_n(0)e^{-\beta t}$.
In the case that $q$ is imaginary, we can then solve system \eqref{eq:vu_matvec} to find
\begin{equation}
\begin{split}
    v_n(t) = 
    \ &e^{-pt} \cos(|q|t) \bigg( v_n(0)
     - s_n(0) \frac{2p - (\beta + 1)}{(p - \beta)^2 - q^2}
     - I \frac{2p - 1}{p^2 - q^2} \bigg) \\
    - &e^{-pt} \frac{1}{|q|} \sin(|q|t) \bigg(
    s_n(0) \frac{p^2 + q^2 -p(\beta + 1) + \beta}{(p - \beta)^2 - q^2}
     + I \frac{p^2 + q^2 - p}{p^2 - q^2}
     + v_n(0)(1 - p) + u_n(0) \bigg) \\
    + &e^{-\beta t} s_n(0) 
    \frac{2p - (\beta + 1)}{(p - \beta)^2 - q^2}     + I \frac{2p - 1}{p^2 - q^2},
\end{split}
\label{eq:v-no-input}
\end{equation}
and
\begin{equation}
    \begin{split}
    u_n(t) = 
    \ &e^{-pt} \cos(|q|t) \bigg(u_n(0) - R \bigg(
    \frac{s_n(0)}{(p - \beta)^2 - q^2} + \frac{I}{p^2 - q^2}\bigg) \bigg) \\
    -& e^{-pt} \frac{1}{|q|} \sin(|q|t) \bigg( R \bigg(
    s_n(0) \frac{p + \beta}{(p - \beta)^2 - q^2} + I \frac{p}{p^2 - q^2} - v_n(0)\bigg)
     + (p - 1) u_n(0) \bigg) \\
    + &e^{-\beta t} \frac{Rs_n(0)}{(p - \beta)^2 - q^2} 
    + \frac{RI}{p^2 - q^2}.
    \end{split}
    \label{eq:u-no-input}
\end{equation}
These equations may be simplified further for practical manipulation using lumped parameters, as discussed in Appendix \ref{app:vu-noinput}.  Note that \eqref{eq:v-no-input}--\eqref{eq:u-no-input} only specify the solutions to \eqref{eq:v_discrete}--\eqref{eq:u_discrete} between firing events; full solutions to the system must be obtained by finding the ordered sequence of global firing times $\{t_{{n_j},{k_j}}\}_{j\in\mathbb{Z}}$ (where $j$ indexes over all firing events of all neurons) and resolving the reset events associated with each such event (in order):
\begin{equation}
    v_{n_j}\big(t_{n_j,{k_j}}^+\big) = v_\text{r}, \quad
    s_{m \neq n_j}\big(t_{n_j,{k_j}}^+ \big) = 
    s_{m}\big(t_{n_j,k_j}^- \big) + \beta\Delta w\big(|x_{n_j}-x_{m}|\big),
\end{equation}
where $g(t^\pm) = \lim_{\epsilon \to 0} g(t \pm \epsilon)$ for some function $g$.
The firing times themselves must be found by solving the nonlinear equation $v_n(t_{{n_j},{k_j}}^-) = v_\text{th}$.

\subsection{Subthreshold resonance}
For certain values of $R$ and $D$, our neuron model demonstrates a substhreshold resonant response, as discussed in more detail in~\cite{Richardson2003, Brunel2003, Rotstein2014}. Briefly, for a sinusoidally-forced reduced model of form
\begin{align}
    \frac{dv}{dt} &= -v -u + \sin(\omega t), \label{eq:forced_v} \\
    \frac{du}{dt} &= Rv - Du, \label{eq:forced_u}
\end{align}
that neglects fire-and-reset mechanics, $v$ has solutions of form $v(t) = \alpha(\omega) \sin(\omega t - \theta(\omega))$.  The system is defined as exhibiting resonance when the amplitude $\alpha$ has a maximum at a forcing frequency $\omega^*/2\pi \neq 0$, called the resonant frequency.
The resonant frequency of our neuron for different values of $R$ and $D$ is plotted in Figure \ref{fig:stability}.2, and compared to the natural frequency of the neuron $|q|/2\pi$ in Figure \ref{fig:stability}.3.  The unplotted (white) regions are areas in which no resonant response occurs.  We see that the resonant region is limited to within the oscillatory regime, with the resonant frequency increases in $R$ and decreases in $D$, and holding a similar value to the natural frequency except near the boundary of the resonant region.

\subsection{Delayed response to transient forcing}
\label{sec:transient}
\noindent As we intend to examine travelling wave solutions, we can prime our intuitions by considering the behaviour of a single neuron as it is exposed to input resembling that of an approaching wave.
We consider a single neuron similar to \eqref{eq:forced_v}--\eqref{eq:forced_u}, differing in that the forcing applied is a transient negative-then-positive input, corresponding to the lateral inhibition followed by local excitation expected from our Mexican hat coupling \eqref{eq:mexican-hat}.  We use the system
\begin{align}
    \frac{dv}{dt} &= -v -u + s(t), \label{eq:trans-v}\\
    \frac{du}{dt} &= Rv - Du, \label{eq:trans-u}\\
    s(t) &= \frac{c}{\sqrt{2\pi}}e^{-\frac{c^2}{2}(t - \tau - 1/c)^2} 
    - \frac{c}{\sqrt{2\pi}}e^{-\frac{c^2}{2}(t - \tau + 1/c)^2}, \label{eq:trans-i}
\end{align}
for some $\tau \gg {1}/{c}$ and initial conditions $\left(v(0),u(0)\right) = \left(0,0\right)$.  The forcing term $s(t)$ is negative until time $t = \tau$, at which point it becomes positive, mimicking the transition from inhibition to excitation as travelling wave of speed $c$ approaches the position of a neuron.  A faster wave gives a smaller Gaussian variance ${1}/{c^2}$, compressing the same input over a shorter time interval.

\begin{figure}
    \centering
    \includegraphics[width=\linewidth]{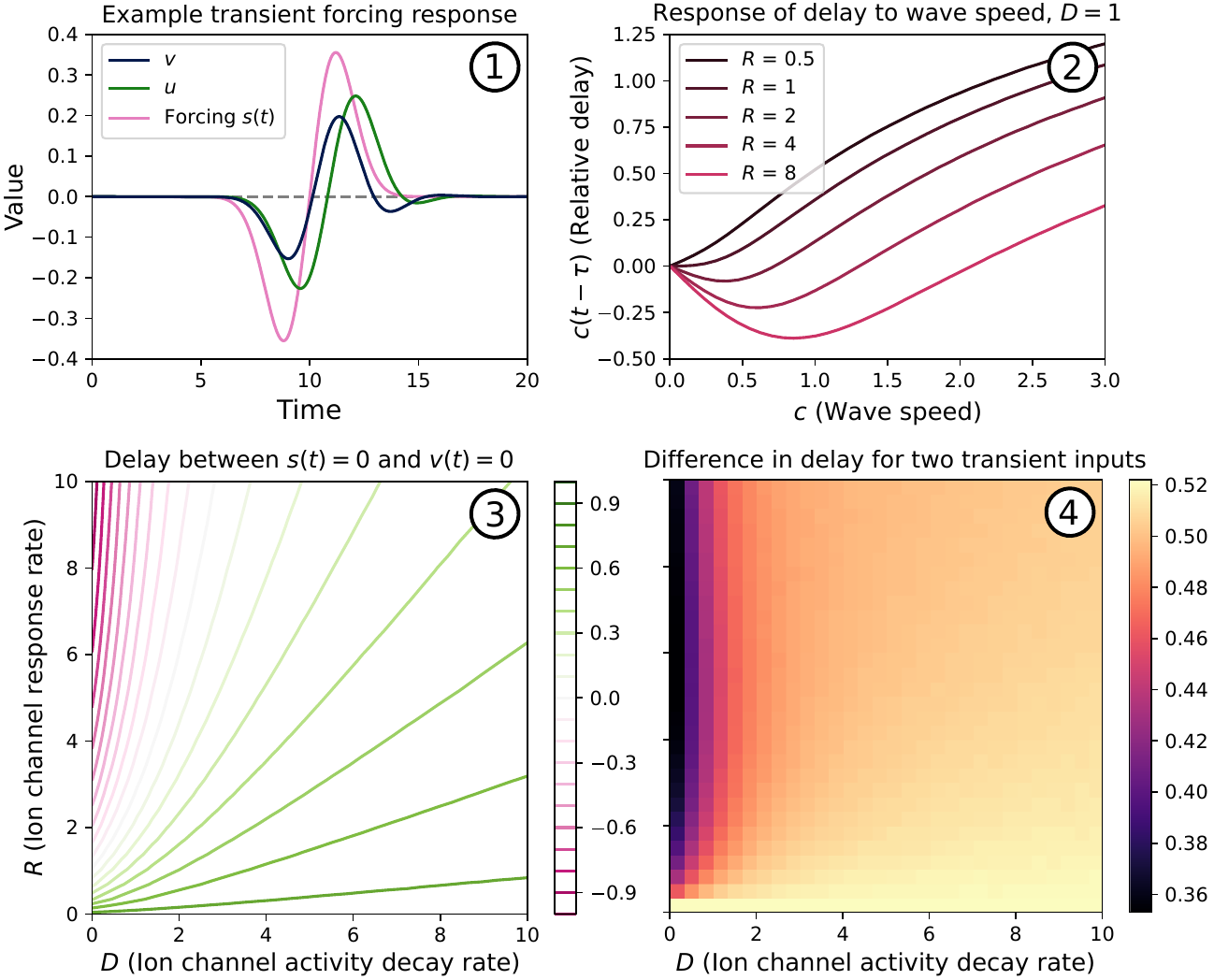}
    \caption{1: Example solution to the system of equations \eqref{eq:trans-v}--\eqref{eq:trans-i} with $R = 2$, $D = 1$, $\sigma = 1$, $\tau = 10$.  2: Change in relative delay $c(t^* - \tau)$ for $D = 1$ as a function of $c$.  3: Level sets of $t^* - \tau$ with $c = 1$ such that $t^* > 0$ is the smallest value that solves $v(t^*) = 0$.  4: Change in $t^* - \tau$ with $c = 1$ when a second input $s$ is introduced offset one time unit after the first.}
    \label{fig:transient-delay}
\end{figure}

Figure \ref{fig:transient-delay}.1 illustrates an example solution to this system in which $v$ falls then rises in response to the external forcing $s$, following some delay.  We quantify this delay by comparing the time between $t = \tau$ at which $s$ reaches $s(t) = 0$ from below with the time $t^*$ at which $v$ reaches $v(t) = 0$ from below.  
Figure \ref{fig:transient-delay}.3 shows the level sets of $t^* - \tau$ across varying values of $R$ and $D$ for $c = 1$, indicating that the delay decreases with $R$ and increases with $D$.  The level sets possess an approximately radial structure, indicating weak dependence upon the timescale of $u$, with the ratio $R/D$ and thus the magnitude of $u$ being the dominant factor.

Figure \ref{fig:transient-delay}.2 shows the change in the relative delay $c(t^* - \tau)$ as $c$ is varied for various values $R$ as $D = 1$ is fixed.  This indicates that a faster timescale of changing input generally leads to a longer delay in the response of $v$ once the delay is expressed in the timescale of the input.  This can be considered the ``distance" that $v$ lags behind $s$, by analogy to our travelling waves.  Exceptions occur at high values of $R$ and low values of $c$, for which the relative delay instead decreases as $c$ increases.
Finally, Figure \ref{fig:transient-delay}.4 shows the increase in delay caused by introducing a second input (modelling a second travelling spike in the wave) of the form 
\begin{equation}
    \frac{dv}{dt} = -v -u +s(t) + s(t - 1),
\end{equation}
where 1 is an offset chosen so that the two inhibitory windows overlap.  We see that there is a largely constant increase in delay across a broad range of values of $R$ and $D$, with the increase being smaller only when $D$ is small.

For a travelling wave, the inputs to a single neuron are driven by the firing pattern of the wave, which is tied directly to the time at which the neuron must fire as determined by the implicit wave equations.  For example, in a one-spike wave it is necessary that $v(t^-) = v_\text{th}$ and the maximum value of $s(t)$ occur at the same time, as the maximum value of $s(t)$ is itself generated at the firing position of the wave.  As this relationship is invariant with respect to parameter choice, a change in one parameter that alters the rate at which the neuron responds to the wave's input must be counterbalanced by a change in speed that restores the original response timing.  We therefore expect that an increase in $R$ will lead to an increase in $c$, and that adding an additional spike to the wave will lead to a decrease in $c$.  This slowing of the wave due to additional spikes has previously been observed in a laterally-inhibited IF model without a slow ion channel \cite{Avitabile2023}, and did not occur in a spiking model with only excitatory connections \cite{Ermentrout1998}.

\section{An event-based numerical algorithm for evolving the network state}
\label{sec:numerics}

\begin{figure}
    \begin{center}
    \includegraphics[width=\linewidth]{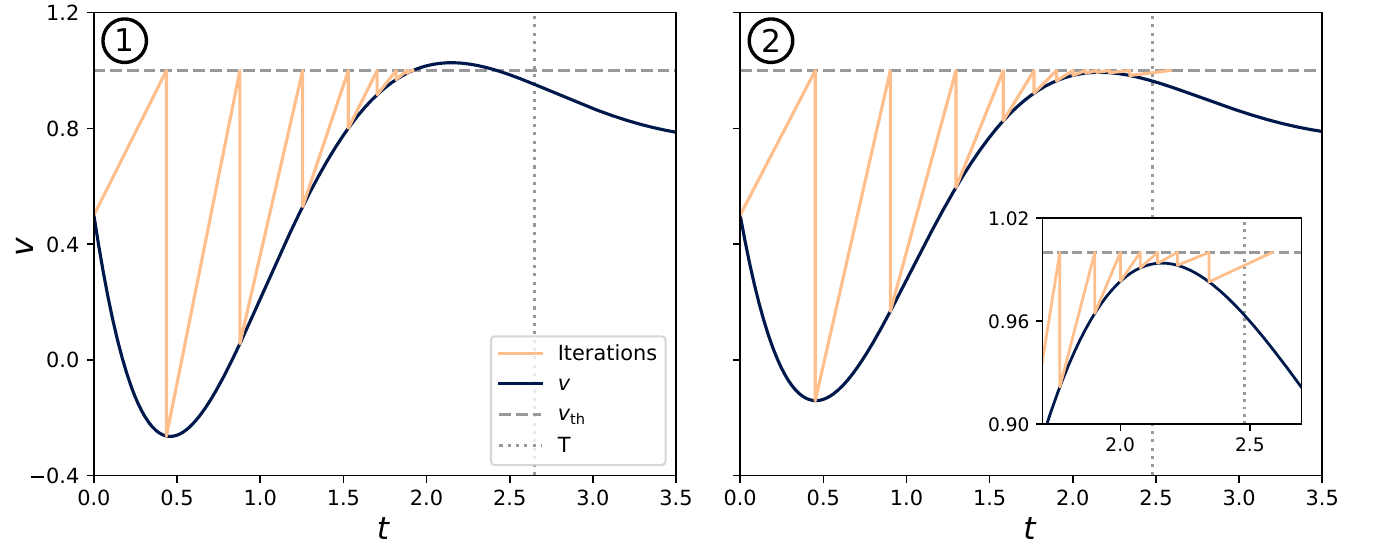}
    \caption{Examples of our root-finding algorithm in use for two example neurons.  The curve gives the value of $v(t)$, while the orange lines trace the progress of the algorithm up to 1: a root, or 2: the end of the possible solution interval $T$.}
    \label{fig:root_finding}
    \end{center}
\end{figure}

To perform numerical simulations of our discrete model \eqref{eq:v_discrete}--\eqref{eq:s_discrete}, we have developed an efficient event-based simulator that evolves the entire network state directly from one (global) firing time to the next in a similar fashion to~\cite{Bonilla-Quintana2017,Engelken2023}.  This approach is detailed further in Appendix \ref{app:numerics}, and we here present only a brief overview of the numerical scheme.

The algorithm comprises three steps. The first exploits the embarrassingly parallel nature of the problem (since our neurons are effectively uncoupled between firing events) to compute the next firing time for each neuron in the network in parallel, or determine that no such time exists.
The second step identifies the smallest firing time across the network and the third step advances the network to this firing time and then applies update rules determined by which neuron is firing.

The first of these steps takes advantage of our ability to find an analytical solution \eqref{eq:v-no-input}--\eqref{eq:u-no-input} to our single-neuron model between firing events.  However, we note to highlight that our analytical solution for $v$ \eqref{eq:v-no-input} does not admit an analytical solution to the inverse problem $v(t^*) = v_\text{th}$, and even determining the existence of a solution to said firing condition is non-trivial.  

To efficiently compute firing times for a given neuron, we exploit the exponentially decaying nature of $v$ to establish an upper bound $T > t^*$ on any firing time (regardless of whether $t^*$ actually exists), then modify the standard Newton--Raphson iterative formula
\begin{equation}
    t_{n+1} = t_{n} + \frac{v_\text{th} - v(t_n)}{m_n}, \quad n = 0,1,2, \dots
\end{equation}
by taking a value $m_n \geq \sup\{v'(t) : t \in [t_n, t_{n+1}]\}$ that overestimates the rate of change in $v$, preventing iterations from overshooting any root $v(t^*) = v_\text{th}$.  In the absence of a root, the iterations will eventually attain either $t_n > T$ or $m_n < 0$, with both cases indicating that no root exists (recalling the condition $v(t_0) < v_\text{th}$).  The sequence $m_n$ is chosen so that it converges to $v'(t)$ near the root, thus recovering the quadratic convergence of the standard Newton--Raphson algorithm as illustrated in Figure \ref{fig:root_finding}. Further details of the simulation algorithm, including the formulae used and pseudocode representations, are included in Appendix \ref{app:numerics}.

\section{Existence and stability of travelling waves}
\label{sec:tw_construction}

\subsection{Construction of travelling waves}
To facilitate our analysis of travelling waves in spiking neuronal networks analytically, it is more convenient to work in the continuum limit as $N \to \infty$ such that $\Delta_x ={2L}/{N} \to 0$.  In such a limit, the firing times $t_{n,k}$ are replaced by the firing time functions $t_k(x)$, which represent the $k$th firing time of a neuron at position $x$~\cite{Avitabile2023}.
Similarly, the locations of the firing neurons at a given time $t$ are captured by the function $X_k(t) := t_k^{-1}(t)$.
We label the pattern of firing described by a single firing position function $X_k$ as a \textit{wave component}, as a travelling wave in which each neuron fires $m$ times includes $m$ such functions in its description.

Taking the continuum limit of our input \eqref{eq:discrete-input} yields
\begin{equation}
    f_\text{in}(x, t) = \sum_{n \in \mathbb{Z}} \int_{\Omega} w(|{x} -{y}|) \delta\big(t - t_n({y})\big) \text{d}{y}.
    \label{eq:input_continuum}
\end{equation}
As in \cite{Avitabile2023}, we assume the the firing time functions are consistently ordered such that $t_{k+1}(x) > t_{k}(x)$ for all $k$ and all $x \in \mathbb{R}$.  Furthermore, we only consider the case in which $t_n(x)$ is monotonic.  Put together, in the continuum limit \eqref{eq:v_discrete}--\eqref{eq:s_discrete} becomes
\begin{align}
    \frac{\partial v}{\partial t} &= I - v - u + s - (v_\text{th} - v_\text{r}) \sum_{k \in \mathbb{Z}} \delta \big(t - t_k(x)\big), \label{eq:v_continuum} \\
    \frac{\partial u}{\partial t} &= Rv - Du, \label{eq:u_continuum} \\
    \frac{\partial s}{\partial t} &= -\beta s + \beta f_\text{in}(x, t).
\label{eq:s_continuum}
\end{align}

\noindent The continuum model \eqref{eq:v_continuum}--\eqref{eq:s_continuum} supports travelling wave solutions with speed $c > 0$ and firing time functions of the form
\begin{equation}
    t_j(x) = \tau_j + \frac{x}{c}, \quad j = 1,2,\dots,m,
\end{equation}
where $\tau_j \geq 0$ are temporal offsets such that $\tau_j - \tau_k$ gives the time between the $j$th and $k$th firing.  The translation invariance of \eqref{eq:v_continuum}--\eqref{eq:s_continuum} allows us to freely choose our origin, so we set $\tau_1 = 0$ for convenience.
Following this constraint, an $m$-spike wave (in which each neuron fires $m$ times) has $m$ unknown variables, $(c, \tau_2, ... , \tau_m)$, and is defined by the $m$ conditions
\begin{align}
    & \tau_1 = 0, \label{eq:wave-cond-1}\\
    & v\left(x,(\tau_j + x/c)^-\right) = v_\text{th}, \quad j=1,\dots,m \label{eq:wave-cond-2}, \\
    & v\big(x, (t + x/c)\big) < v_\text{th}, \quad
    \forall t \notin \{\tau_1, ..., \tau_m\},\ x \in \mathbb{R}. \label{eq:wave-cond-3}
\end{align}
Without loss of generality we order the firing offsets so that $\tau_1 < \tau_2 < \dots < t_m$.
Note that condition \eqref{eq:wave-cond-3} ensures that the wave contains exactly $m$ spikes since violation of this condition would introduce additional firing events and thus change the spike count.

Each of the $m$ firing events provide a contribution to the total input to each neuron of the form \eqref{eq:input_continuum}, which can be summed linearly to give a total input
\begin{equation}
    f_\text{in}(x, t) 
    = \sum_{j=1}^m \int_{-\infty}^\infty
    w\big( |x - x'|\big) \delta \big(t - (\tau_j + x'/c) \big) dx'.
    \label{eq:fin_wave}
\end{equation}
Upon moving to the co-moving frame with $\xi = t-x/c$, \eqref{eq:v_continuum}--\eqref{eq:s_continuum} becomes
\begin{align}
    \frac{\partial v}{\partial \xi} &= I - v - u + s - (v_\text{th} - v_\text{r}) \sum_{j=1}^m \delta (\xi - \tau_j), \label{eq:v_comoving}\\
    \frac{\partial u}{\partial \xi} &= Rv - Du, \label{eq:u_comoving}\\
    \frac{\partial s}{\partial \xi} &= -\beta s + \beta f_\text{in}(\xi), \qquad f_\text{in}(\xi) = c \sum_{j=1}^m w\big( c(\xi - \tau_j) \big).
    \label{eq:s_comoving}
\end{align}
Travelling waves are then time-invariant solutions to \eqref{eq:v_comoving}--\eqref{eq:s_comoving}, so that for $\mathbf{v}=(v,u)^\top$ we have $\mathbf{v}(\xi,t) = \mathbf{v(\xi)}$ where
\begin{equation}
\begin{split}
    \textbf{v}(\xi) =&\ I e^{\xi \textbf{M}}  \int_{-\infty}^\xi
    e^{-\zeta\textbf{M}} \begin{pmatrix} 1 \\ 0 \end{pmatrix} d\zeta
    + \sum_{j=1}^m  \beta e^{\xi \textbf{M}} \int_{-\infty}^\xi
    e^{-\zeta \textbf{M}} \begin{pmatrix} 1 \\ 0 \end{pmatrix}
    e^{-\beta \zeta } \int_{-\infty}^\zeta e^{\beta r} 
    w\big( c(r - \tau_j) \big) c\ dr\ d\zeta\\
    & \qquad - (v_\text{th} - v_\text{r}) \sum_{j=1}^m e^{(\xi - \tau_j )\mathbf{M}} 
    \begin{pmatrix} 1 \\ 0 \end{pmatrix}
    \Theta(\xi - \tau_j),
\end{split}
\label{eq:comoving_soln}
\end{equation}
where $\Theta$ is the Heaviside step function.  
We can the apply the Newton--Raphson scheme to the $m$ conditions $v(\tau_m^-) = v_\text{th}$ to find the unknowns $(c, \tau_2, ..., \tau_m)$ that specify wave solutions with the desired number of firing events.
Note that for each $m$, we expect to find multiple distinct wave solutions existing for the same parameter values, as seen in previous work on the in the case $R=0$ \cite{Avitabile2023}.  Examples can seen below for $m = 1$ and $m = 2$ in Figure \ref{fig:PAC-solutions-2spike}.

\subsection{Wave perturbations and stability}

\noindent Linear stability of the travelling waves specified by \eqref{eq:comoving_soln} can be examined by studying the response of system \eqref{eq:v_continuum}--\eqref{eq:s_continuum} to perturbations in the interval between firing events, since neurons interact solely through these firing events.  This approach to stability analysis is detailed in~\cite{Avitabile2023}, which in turn extends \cite{Bressloff2000}.  Here we recapitulate the key results and refer the interested reader to the aforementioned paper.

An $m$-spike wave comprises $m$ wave components, each moving at speed $c$ satisfying
\begin{equation}
    v\big( X_j(t),\ t^- \big) = v_\text{th}, \quad j \in \{1, ..., m\},
\end{equation}
with the unperturbed firing positions given by $X_j(t) = c(t - \tau_j)$.  Following~\cite{Bressloff2000}, we introduce a perturbation of the form
\begin{equation}
    \tilde{X}_j(t) = c(t - \tau_j) + \epsilon \phi_j(t),\quad
    \phi_j(t) = \text{Re}\big(\Phi_j e^{\lambda t}\big) + o(\epsilon),
\end{equation}
for $\lambda, \Phi_j \in \mathbb{C}$ and $\phi'_j(t) > -c$, $\forall t$.  This introduces variability in the speed of each wave component, which may be growing or decaying, and may be oscillatory, but never causes them to reverse direction.  For brevity, in subsequent equations we will simply write $\phi_j(t) = \Phi_j e^{\lambda t} + o(\epsilon)$ while still only considering the real part implicitly.

The input to each wave under the perturbed firing times is given by
\begin{equation}
    f_\text{in}(x, t) = c \sum_{j=1}^m 
    w\big(x - c(t - \tau_j) + \epsilon\phi_j(t)\big).
\end{equation}
Following from \eqref{eq:comoving_soln}, we can write $m$ equations describing the behaviour of variables at our perturbed firing events as
\begin{equation}
\begin{split}
    \tilde{\textbf{v}} \big( X_k(t) + \epsilon \phi_k(t), t \big) =&\ I e^{t\textbf{M}}  \int_{-\infty}^t 
    e^{-t'\textbf{M}} \begin{pmatrix} 1 \\ 0 \end{pmatrix} dt'
    + \sum_{j=1}^m  \beta e^{t\textbf{M}} \int_{-\infty}^t 
    e^{-t'\textbf{M}} \begin{pmatrix} 1 \\ 0 \end{pmatrix} S_{jk}(t') dt'\\
    & \qquad -(v_\text{th} - v_\text{r}) \sum_{j=1}^{k-1} 
    e^{\left(t - \tilde{X}_j^{-1}\left(\tilde{X}_k(t)\right)\right)\mathbf{M}},
\end{split}
\end{equation}
with $d_{jk} := \tau_j - \tau_k$, noting that have already evaluated the Heaviside functions since we know which firing event $k$ we are following. 
The $S_{jk}$ terms expanded to first order in $\epsilon$ are given by
\begin{equation}
\begin{split}
    S_{jk}(t') =&\ e^{-\beta t'} \int_{-\infty}^{t'} e^{\beta r}
    w\big( c(r - t) - cd_{jk}  \big) \cdot c\ dr - \epsilon \left( 
    w\big( c(t' - t) - cd_{jk} \big)
    \big(\Phi_j e^{\lambda t'} - \Phi_k e^{\lambda t} \big) \right) \\
    & \qquad + \epsilon \beta e^{-\beta t'} \int_{-\infty}^{t'} e^{\beta r}
    w\big( c(r - t) - cd_{jk} \big) 
    \big(\Phi_j e^{\lambda r} - \Phi_k e^{\lambda t} \big) dr
    + O\big(\epsilon^2 \big).
\end{split}
\end{equation}
To complete our linearisation, we make use of
\begin{equation}
    \tilde{X}_j^{-1}(x) = X_j^{-1}(x) - \frac{\epsilon}{c} \phi_j\left(X_j^{-1}(x)\right) + O\left(\epsilon^2\right),
\end{equation}
allowing us to rewrite the firing-reset exponentials as
\begin{equation}
    \exp \left(\left(t - \tilde{X}_j^{-1}\big(\tilde{X}_k(t)\big)\right)\mathbf{M} \right)
    = e^{-d_{jk} \mathbf{M}}
    \left(\mathbf{I} + \frac{\epsilon}{c} e^{\lambda t}
    \left(\Phi_j e^{\lambda d_{jk}} - \Phi_k\right) \mathbf{M} \right)
    + O\left(\epsilon^2\right).
\end{equation}
Finally, we arrive at the equation for the perturbed $v$ expanded to first order in $\epsilon$:
\begin{equation}
    \tilde{v} \big( X_k(t) + \epsilon \phi_k(t), t \big) 
    = v_\text{th}
    + \epsilon e^{\lambda t} \sum_{j=1}^m \big( \Phi_j F_{jk}(\lambda) - \Phi_k F_{jk}(0) \big)
    + O\left(\epsilon^2\right),
\label{eq:wave_perturb_expanded}
\end{equation}
where
\begin{equation}
\begin{split}
    F_{jk}(\lambda) 
    = (1, 0) \Bigg(& \beta 
    \int_{-\infty}^0 e^{-t'\textbf{M}} \begin{pmatrix} 1 \\ 0 \end{pmatrix} 
    \bigg( w\big( ct' - cd_{jk} \big)  e^{\lambda t'}
    - \beta e^{-\beta t'} \int_{-\infty}^{t'} e^{\beta r} 
    w\big( cr - cd_{jk} \big) e^{\lambda r} dr\bigg) dt'\\
    &- \mathbbm{1}_{j < k}  (v_\text{th} - v_\text{r}) \frac{1}{c}  
    e^{-d_{jk} \mathbf{M}} \mathbf{M} \begin{pmatrix} 1 \\ 0 \end{pmatrix} e^{\lambda d_{jk}}
     \Bigg).
\end{split}
\end{equation}
Since these equations describe the perturbed wave's firing events, we require $v \big( X_k(t) + \epsilon \phi_k(t), t \big) = v_\text{th}$; that is, the spatiotemporal profile for $v$ at the perturbed firing times must be at threshold.  Given \eqref{eq:wave_perturb_expanded}, this means we require the order-$\epsilon$ and above terms to evaluate as zero. 
Considering only the order-$\epsilon$ terms and writing $\mathbf{\Phi} = (\Phi_1, ..., \Phi_m)^\top$, this gives us the condition
\begin{equation} 
    \epsilon e^{\lambda t} \big( \mathbf{F}(\lambda) - \mathbf{G} \big) \mathbf{\Phi} = 0 \quad \forall t,
    \label{eq:stab_mat_eq}
\end{equation}
where $\mathbf{F}$ is an $m$-by-$m$ matrix and $\mathbf{G}$ is a diagonal matrix with entries
\begin{equation}
\left[\textbf{F}(\lambda)\right]_{jk} = F_{jk}(\lambda), \quad \textbf{G}_{kk} = \sum_{j=1}^m F_{jk}(0).
\end{equation}
For \eqref{eq:stab_mat_eq} to have non-trivial solutions when $\mathbf{\Phi} \neq \mathbf{0}$, we require $\ker\big( \mathbf{F}(\lambda) - \mathbf{G} \big) \neq \{\mathbf{0}\}$, which in turn requires a $\lambda$ such that
\begin{equation}
    \det \big( \mathbf{F}(\lambda) - \mathbf{G} \big) = 0.
\label{eq:stability_cond}
\end{equation}
The collection of all such $\lambda$ gives the point spectrum of the computed wave.  While this is not a sufficient tool to find all such $\lambda$ or otherwise prove their (non)existence within the positive-real half of the complex plane, it enables us to find precise values for specific roots and to determine the presence (or lack thereof) of roots in finite regions, which can give a strong indicator of linear stability that can then be compared to numerical simulations.
Note that there will always be an eigenvalue $\lambda = 0$ due to the translation invariance of \eqref{eq:v_continuum}--\eqref{eq:s_continuum}. If the real parts of all other $\lambda$ are negative, then the perturbation will decay and hence the wave is stable. If $\text{Re}(\lambda) > 0$ for some $\lambda$ then the perturbation will grow and the wave is unstable.

\section{Dependence of wave solutions on model parameters}
\label{sec:results}
With equations for travelling wave construction \eqref{eq:comoving_soln} and stability \eqref{eq:stability_cond} in hand, we now proceed to investigate how these waves vary as we move through parameter space.
We shall consider variation with respect to the parameters $R$, $D$ for the ion channel variable $u$ and also with respect to the synaptic timescale $\beta$ for comparison with results in~\cite{Avitabile2023}.
To do this, we use the pseudo-arclength continuation method~\cite{Dahlke2024,Laing2014} to track wave solution branches after having found suitable initial guesses.
Our base set of parameters is indicated in Table~\ref{table:PAC_params}, with non-default parameter values indicated in the main text and figure captions where appropriate.
We remark that we allow our algorithm to follow non-admissible solutions, i.e., those that violate \eqref{eq:wave-cond-3}, since these \textit{virtual invariant profiles} still play a role in shaping diagrams~\cite{diBernardoBook}, particularly when their spectrum lies entirely in the left complex plane (thus being ``stable"), since these attract nearby solutions.
These inadmissible solution branches, which appear following passage of a admissible branch through a grazing bifurcation where $v(\tau_\text{graze}) = v_\text{th}$, $v_\xi(\tau_\text{graze}) = 0$ for $\tau_\text{graze} \notin \{\tau_1, \dots, \tau_m\}$, are plotted in a lighter shade compared to the admissible solutions to distinguish between the two.

\begin{table}
\begin{center}
    \begin{tabular}{|c|c|c|}
        \hline
        \textbf{Parameter} & \textbf{Symbol} & \textbf{Value} \\
        \hline
        Ion channel activity decay rate & $D$ & 1 \\
        \hline
        Synaptic response rate & $\beta$ & 6 \\
        \hline
        Resting voltage & $v_\text{rest}$ & 0.9 \\
        \hline
        Excitatory/inhibitory signal strength scale & $A, B$ & 2 \\
        \hline
        Excitatory signal range scale & $a$ & 1 \\
        \hline
        Inhibitory signal range scale & $b$ & 2 \\
        \hline
        Number of neurons simulated & $N$ & 2000 \\
        \hline
        Domain length simulated & $2L$ & 20 \\
        \hline
    \end{tabular}
    \caption{Default parameter values for wave solutions and simulations.}
    \label{table:PAC_params}
\end{center}
\end{table}

\subsection{Wave solutions under variation of ion channel parameters}

Figure \ref{fig:PAC-solutions-2spike} plots the wave speeds $c$ of one- and two-spike solutions (gold and blue curves, respectively) as a function of the ion channel response rate $R$, taking fixed $D = 1$, $\beta = 6$, along with selected example profiles of $\left(v(\xi),u(\xi),s(\xi)\right)$ and raster plots of corresponding simulations at various points along the branches.
Simulations were run until either 4000 firing events were found or no further neurons would fire.
We observe the existence of two distinct branches of stable two-spike wave solutions; a slower wave (dark blue, Example 3) with a monotonic inter-spike rise in $v$ and single peak in $s$, and a faster wave (light blue, Example 1) with a longer inter-spike period $\tau_2$ and two peaks in $s$ associate with each spike.  The slower branch terminates at a grazing bifurcation corresponding to the maximum in $v$ after the second firing event attaining $v_\text{th}$.  The faster branch instead terminates at a fold bifurcation, after which the waves become unstable (see Example 2).
Here the unstable wave solution possesses a small real eigenvalue, and under simulation its wave components gradually drift apart as it converges to the stable solution of Example 1.

We refer to wave solutions such as the slower stable wave as \textit{atomic} waves, in reference to the fact all the firing events are indivisibly coupled within a single excitatory peak.  This is in contrast to the faster stable wave in which the wider spacing between the firing events and similarity in speed to the one-spike solution suggest that it can be considered a composite of two one-spike waves that are weakly coupled, as is often seen in excitable systems~\cite{seidelCoherentPulseInteractions2025}.

\begin{figure}
    \centering
    \includegraphics[width=\textwidth]{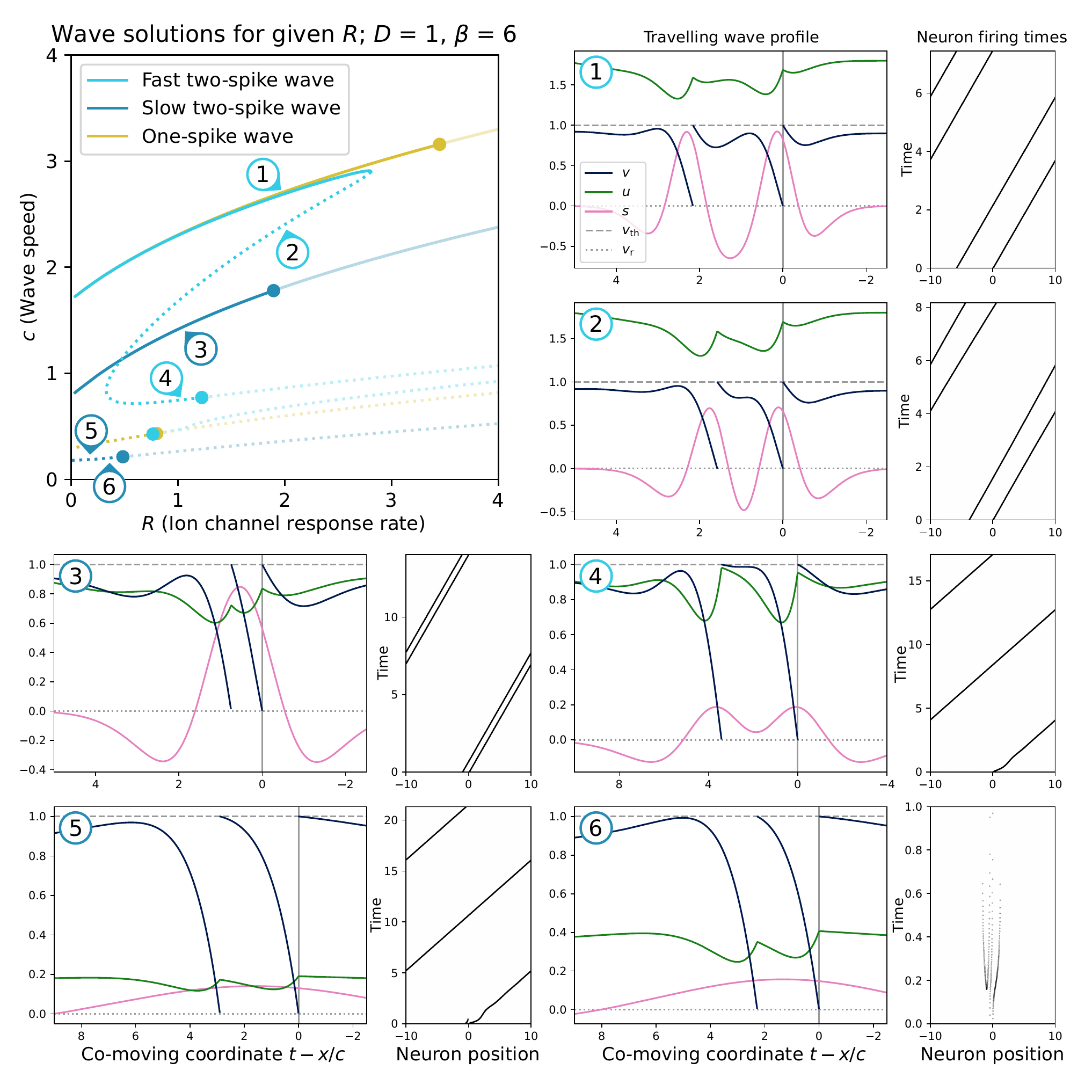}
    \caption[Wave solution bifurcation diagram in $R$.]{Wave speeds $c$ of solutions for varying values of $R$.  Blue curves represent solutions for two-spike waves, and gold curves represent solutions for one-spike waves.  Solid curves are stable, while dotted curves are unstable.  Faded curves are inadmissible solutions that exist beyond a grazing bifurcation (marked with dots).  Six example solutions are presented as travelling wave profiles in co-moving coordinates $t - x/c$, accompanied by the raster plots of firing events found through a numerical simulation that took the the wave profile as its initial conditions.}
    \label{fig:PAC-solutions-2spike}
\end{figure}

Other unstable wave profiles are shown in Examples 4, 5, and 6.
Examples 4 and 5 show the near-immediate loss of a spike, producing a one-spike wave after some oscillation.  Example 6 is on the same branch as Example 5 but closer to the grazing bifurcation; as such, its instability causes it to contact the firing threshold an additional time.
The high values of $v$ either side of the maximum cause additional firing events to propagate in both directions, disrupting the wave and sending the network into a quiescent state.  This initial symmetry is an established problem in the spontaneous generation of travelling waves \cite{Wiener1946}; for a graze to increase the spike count $m$ of a wave, the existing wave has to break the symmetry of the graze without the graze destroying the wave in turn.

For all types of wave presented in Figure \ref{fig:PAC-solutions-2spike}, we observe that the wave speed increases as $R$ increases.  The atomic two-spike wave is slower than the (trivially atomic) one-spike wave, and the difference between the two remains approximately constant across values of $R$.  Both of these results correspond to the intuition generated in the toy model of transient forcing in Sect.~\ref{sec:transient}. In Figure \ref{fig:speed_comp}, we observe the difference in speed between the one-spike wave and the atomic $m$-spike waves up to $m = 10$, seeing that the near-constant difference does not hold for higher values of $m$.  We also note that as $m$ increases, the graze bifurcation occurs for smaller values of $R$, while at $m=3$, a Hopf bifurcation appears near $R=0$ and moves up to higher values of $R$. At $m=7$ the bifurcations cross over and the region of stable admissible solutions disappears entirely.

\begin{figure}
    \centering
    \includegraphics[width=\linewidth]{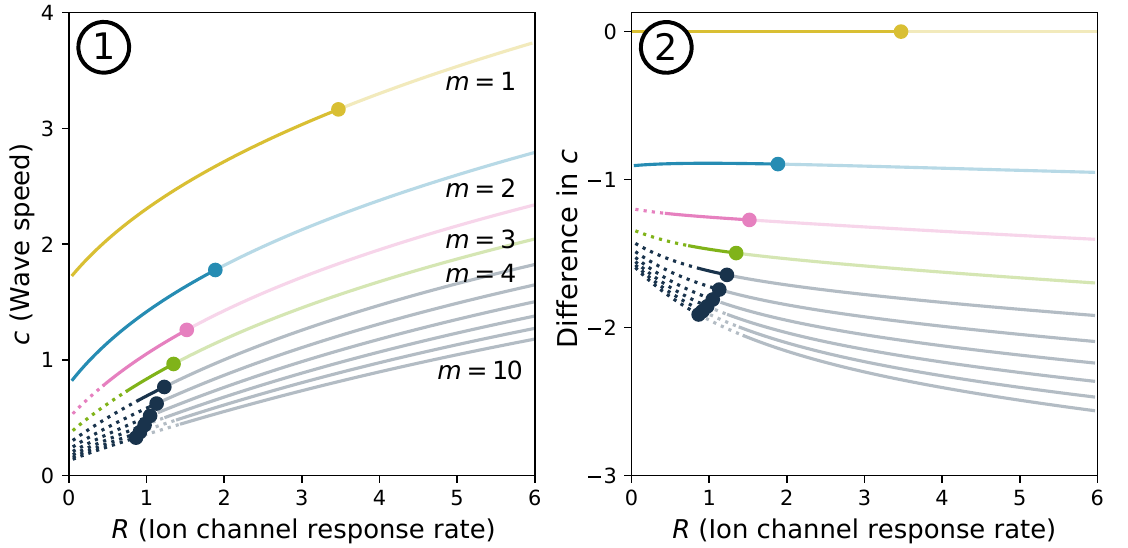}
    \caption{Wave speeds for atomic waves up to 10 spikes.  Curves are ordered with $m=1$ at the top and $m=10$ at the bottom.  Dotted curve sections are unstable, and faded curve sections are inadmissible due to a grazing bifurcation.  1: Absolute speeds.  2: Speeds relative to the one-spike reference wave.}
    \label{fig:speed_comp}
\end{figure}

The weakly-coupled fast two-spike wave is only the first such two-spike solution of its type; there are further weakly-coupled solution branches with higher values of the inter-spike time $\tau_2$, with speeds near-identical to that of the one-spike wave.  As such, they have been omitted from Figure \ref{fig:PAC-solutions-2spike} and instead displayed  alongside the low-$\tau_2$ branch in Figure \ref{fig:waves_nat_freq} for $D = 0.85$, $D=0.88$ and $D=1$.

The values of $\tau_2$ along the solution branches shown closely follow multiples of the natural period of the single neuron model, except in the region around $R = 5$, where each branch (aside from that with the lowest value $\tau_2$) transitions to a multiple of the natural period one integer step lower or higher, determined by a lower or higher value of $D$, respectively, as shown by the cases $D = 0.85$ and $D = 1$.  Between these values of $D$, a rearrangement occurs in which the branches reconfigure as folds near $R=5$ and exchange connections; the case $R=0.88$ gives a snapshot of this process.  The branches retain a stable-unstable alternating pattern, with stable solutions at odd or even multiples of the natural period for $R < 5$ or $R > 5$, respectively.  For our parameter, choices the region of reconfiguration $R \approx 5$ is beyond the position of the graze bifurcations, rendering the solutions in said region inadmissible, however, it is likely that such solutions may be admissible in other variants of our model.

\begin{figure}
    \centering
    \includegraphics[width=\linewidth]{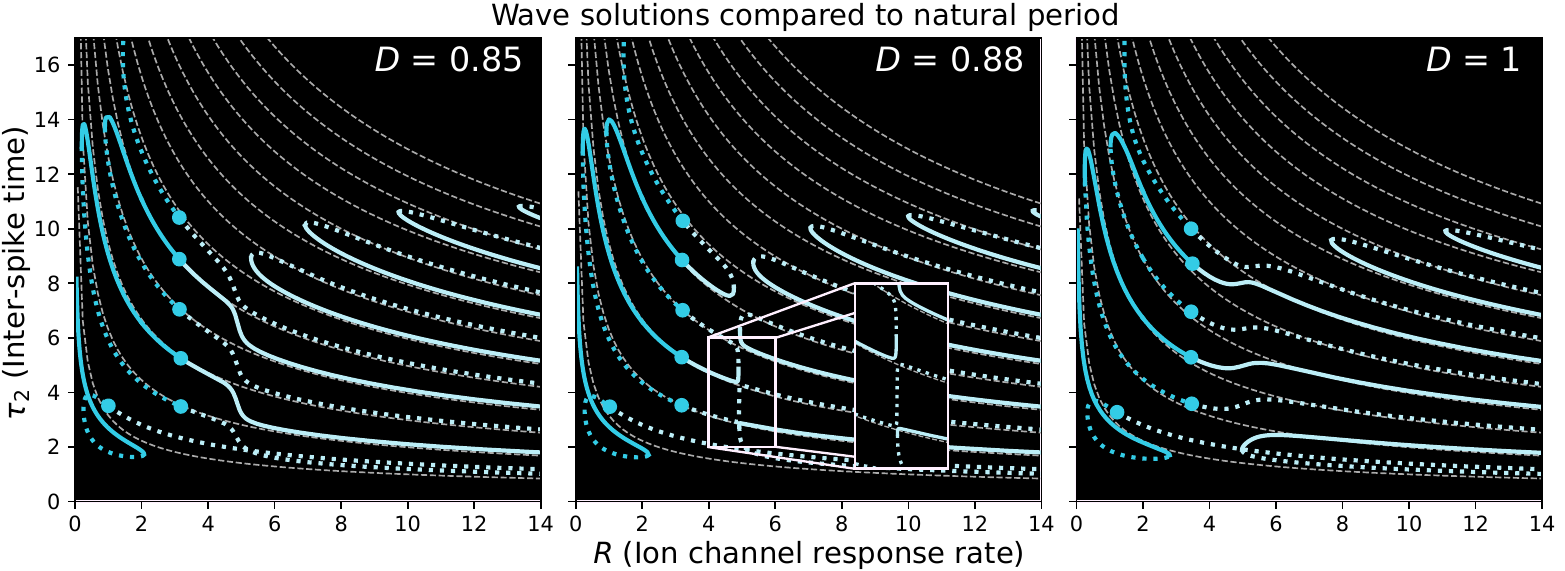}
    \caption[Inter-spike time of fast two-spike waves compared to the natural period of the neuron.]{Inter-spike time $\tau_2$ of fast weakly-coupled two-spike waves compared to integer multiples of the natural period $|q|/2\pi$ (grey dashed lines) as functions of $R$, for fixed values $D=0.85$, $D=0.88$, and $D=1$.  Other parameters are as in Table \ref{table:PAC_params}.  Blue lines and dots mark wave solutions and grazes according to the scheme set out for Figure \ref{fig:PAC-solutions-2spike}.  An inset for $D=0.88$ shows that only one higher-$\tau_2$ branch (unstable, moving from the second to the third multiple) continues across the transitional zone rather than folding back.}
    \label{fig:waves_nat_freq}
\end{figure}

\begin{figure}
    \centering
    \includegraphics[width=0.5\textwidth]{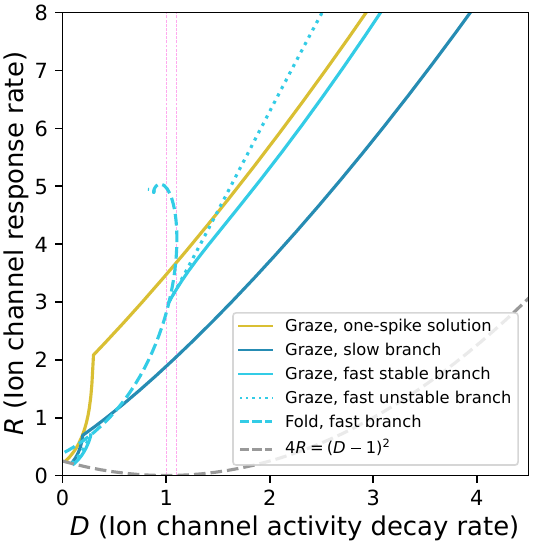}
    \caption{Two-parameter bifurcation diagram under variation of $R$ and $D$ for the graze and fold bifurcations along the stable one- and two-spike wave solutions seen in Figures \ref{fig:PAC-solutions-2spike} and \ref{fig:diptych_R} (coloured accordingly). Pink lines correspond to the values $D=1$ and $D=1.1$ presented in said figures.  Only bifurcations of stable solutions are presented, with the exception of the unstable branch that connects to the fast stable branch via the fold bifurcation.  The area below each bifurcation curve is the region in which the corresponding solution exists.  The limit $4R = (D - 1)^2$ is also shown to indicate the boundary at which the single-neuron solutions $\mathbf{v}$ for \eqref{eq:vu_matvec} no longer take a trigonometric form.  The curve tracking the fold bifurcation is terminated after passing through the reorganisation noted in Figure \ref{fig:waves_nat_freq}.}
    \label{fig:RD_diagram}
\end{figure}

To more completely explore the changes in bifurcations as $D$ varies, Figure \ref{fig:RD_diagram} presents a two-parameter bifurcation diagram in $R$ and $D$.  Curves represent the values of $R$ and $D$ at which a fold (dashed line) or graze (solid or dotted line) exists; for clarity, we focus only the bifurcations associated with stable solution, as well as the graze on the unstable side of the fold of the fast two-spike wave (dotted line).  The bifurcation curves are coloured according to their respective solution branch colours in Figure \ref{fig:PAC-solutions-2spike}.  The values $D=1$ and $D=1.1$ are marked with pink lines to indicate the bifurcations corresponding to the plots in Figure \ref{fig:diptych_R}.
In each case, the bifurcation curve (or sequence of fold and graze curves in the case of the fast stable two-spike solutions) divides the parameter space into the lower-right region (low $R$, high $D$) in which the corresponding admissible solution exists, and the upper-left region (high $R$, low $D$) in which the solution does not exist.  This indicates a broad pattern that a strong adaptation variable $u$ eliminates these wave solutions, while changing the timescale has relatively little impact.  

Grazing bifurcations on the fast two-spike branch exist in two distinct regions of the diagram, with the fold bifurcation splitting them.  The one-spike and slow two-spike solutions do not pass through a fold bifurcation, but do exhibit two distinct sections to their graze-tracking curves, which join at a sharp corner.
Each section corresponds to a graze produced by a different local maximum in the travelling wave profile after the last firing event grazing the firing threshold $v_\text{th}$.
At higher values of $R$, the first post-firing maximum grazes the firing threshold, but at lower values of $R$ it is instead the second maximum.  The point at which the curve segments meet in the $(R,D)$ plane projection gives the parameter values at which both maxima graze the threshold simultaneously.  By enabling the generation of this double graze, the subthreshold oscillations of the neuron's voltage produce an organising centre in the bifurcation diagram.  This double grazing point is noted under the framework of \textcite{Kowalczyk2006} as a Type III codimension-two grazing bifurcation, with our system as a whole considered a Class A discontinuous system.

\subsection{Simulations across distinct bifurcation scenarios under variation of the ion channel response rate}
\begin{figure}
    \centering
    \includegraphics[width=0.95\textwidth, trim=0 0 0 0]{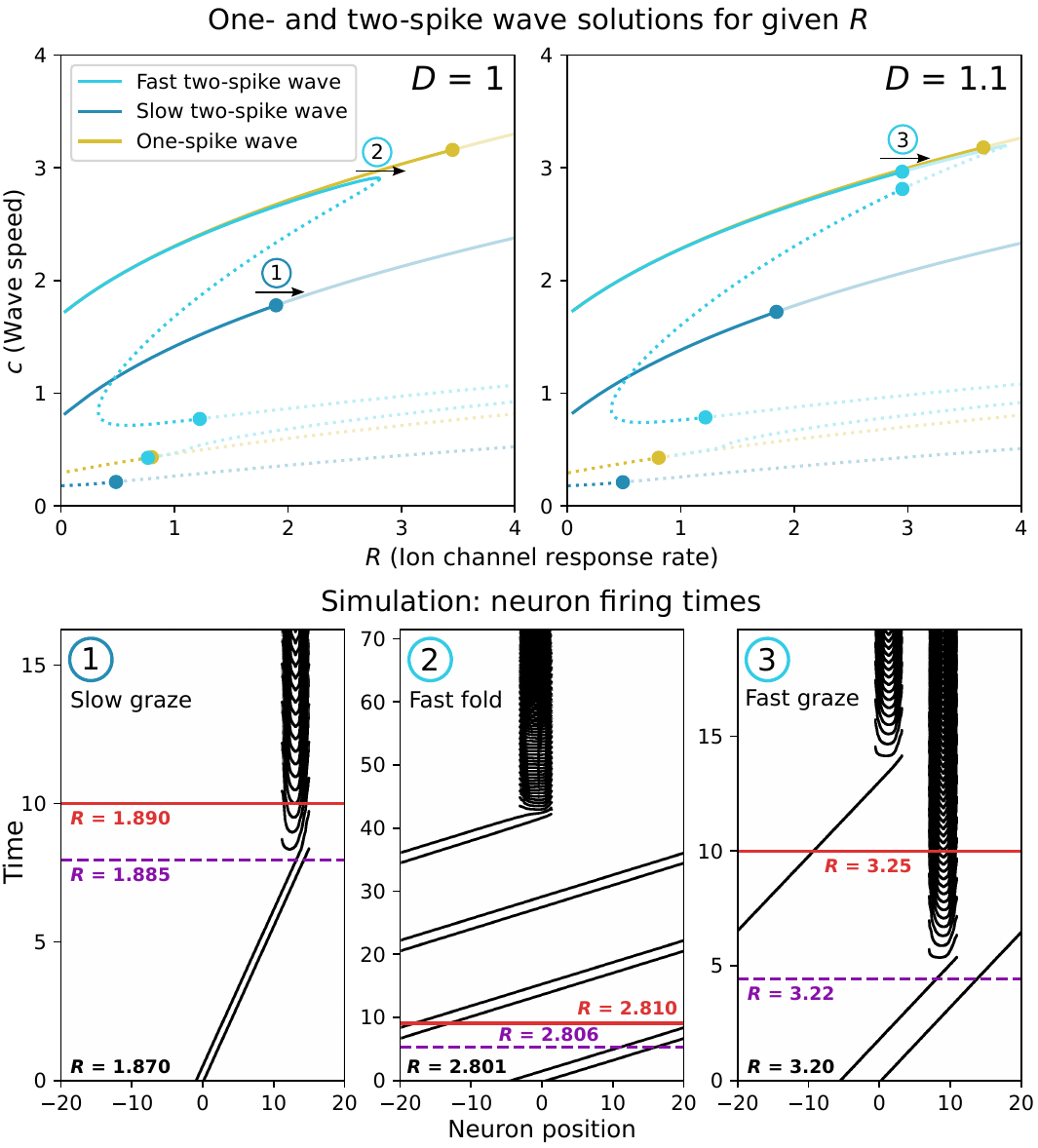}
    \caption[$R$ bifurcation diagrams for $D=1$ and $D=1.1$]{Upper: Bifurcation diagrams for one- and two-spike wave solutions as $R$ varies, shown for $D = 1$ (left) and $D = 1.1$ (right).  As $D$ increases the position of the fold bifurcation moves to the right, allowing grazing bifurcations (marked with a dot) on both forks of the fold to appear.  Lower: Simulations of stable waves as $R$ is linearly increased across a bifurcation.  Red lines show the time at which parameter variation stops, and purple dashed lines show the bifurcation value predicted in the continuum model.  1: Simulation across the slow two-spike branch's graze; $D = 1$, 20000 neurons.  2: Simulation across the fast two-spike branch's fold; $D = 1$, 4000 neurons.  3: Simulation across the fast two-spike branch's graze; $D = 1.1$, 20000 neurons.  }
    \label{fig:diptych_R}
\end{figure}

\noindent As we move from $D=1$ to $D=1.1$, we see in Figure \ref{fig:diptych_R} that the fold on the fast two-spike branch moves to a higher value of $R$, but grazes also emerge near the fold, rendering solutions at the fold itself inadmissible.
We investigate the impact of this change in bifurcation structure via simulation involving slow variation of $R$ according to
\begin{equation}
    R(t) = R_0 + \delta_R \min\{t, t_\text{fin}\},
\end{equation}
for some initial value $R_0$ at initial time t = 0, constant rate of change $\delta_R$, and final value $R_\text{fin} = R_0 + \delta_R t_\text{fin}$.  The insets on Figure \ref{fig:diptych_R} show three such simulations.  Figure \ref{fig:diptych_R}.1 simulates across the graze on the slow two-spike solution branch with $D = 1$, Figure \ref{fig:diptych_R}.2 simulates across the fold on the fast two-spike solution branch with $D = 1$, and Figure \ref{fig:diptych_R}.3 simulates across the graze that appears before the fold on the fast two-spike solution branch with $D=1.1$.  All three simulations were carried out on a ring of neurons, starting from initial conditions representing an unperturbed travelling wave state.  The value $R_0$ are labelled in black, and the time $t_\text{fin}$ and the time at which $R$ reaches the continuum-model bifurcation value are marked with a solid red and dashed purple line, respectively.  In all three cases, we note a transition from a wave to a stationary bump, but the speed of transition and number of bumps varies.

Figure \ref{fig:diptych_R}.1 shows that pushing the fast stable two-spike wave over the fold bifurcation with $D=1$ results in a long delay before the wave transitions to a stationary bump, in accordance with the slow dynamics described by the small eigenvalues near the fold and the expected delayed bifurcation phenomenon under our slow parameter variation~\cite{tzouSlowlyVaryingControl2015}. 
We remark that the resulting bump exhibits near symmetric firing around its centre, with an emerging chevron pattern similar to that observed in~\cite{Coombes2025,Chow2006b}.
Figure \ref{fig:diptych_R}.2 and \ref{fig:diptych_R}.3 instead show the result of pushing a wave over a graze bifurcation.  In both cases, the transition is near-immediate (with some discrepancy due to both the discretisation and the rate of change $\delta_R$), and first marked by the onset of firing expanding outwards from the graze point. With the strong coupling of the slow two-spike wave, the entire wave transitions to a bump at the graze, while the weak coupling of the fast two-spike wave means only the second wave component forms a bump at first, with the first wave component continuing until it loops around the ring and is blocked by the lateral inhibition from the first bump, triggering a transition into a second bump.

\subsection{Wave solutions under variation of the synaptic timescale}
\begin{figure}
    \centering
    \includegraphics[width=\textwidth]{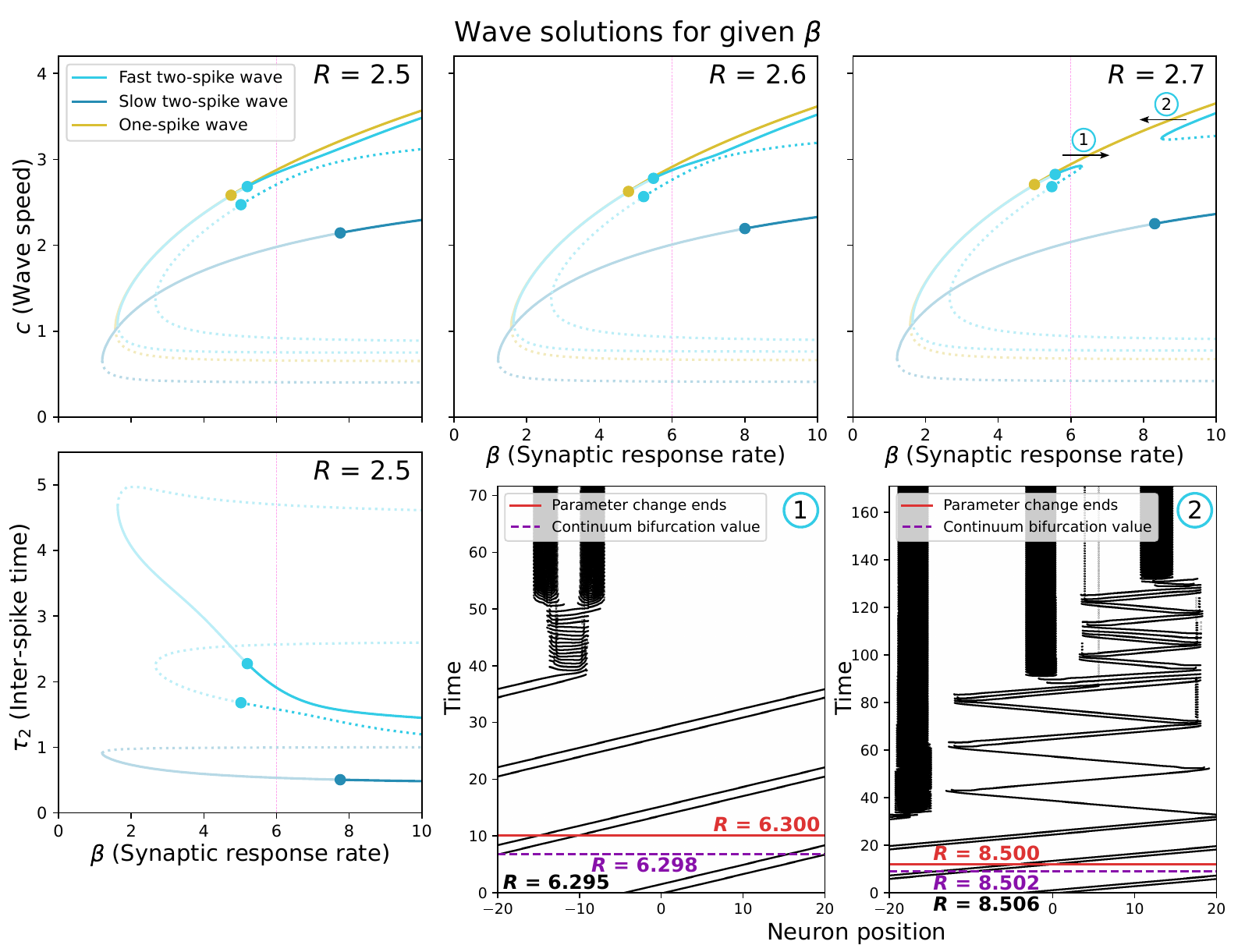}
    \caption[Bifurcation diagrams in $\beta$]{Bifurcation diagrams for one- and two-spike wave solutions as $\beta$ varies, shown for $R = 2.5$, $R = 2.6$ and $R = 2.7$.  Between $R = 2.6$ and $R = 2.7$ two solution branches change connections through a fold bifurcation, producing a gap in which wave solutions cease to exist.  Branches are coloured and annotated as in Figure \ref{fig:PAC-solutions-2spike}, with the pink dashed line marking the value $\beta = 6$ corresponding to the plot in said figure. Numbered panels show simulations on 4,000 neurons as $\beta$ is linearly changed across the fold bifurcations that appear in $R = 2.7$, as in Figure \ref{fig:diptych_R}.}
    \label{fig:triptych-beta}
\end{figure}

\noindent We now examine wave solutions under variation of parameter $\beta$, as illustrated in Figure \ref{fig:triptych-beta}, which shows one- and two-spike wave solutions as functions of $\beta$ for three values of $R =2.5$, $2.6$, $2.7$. 
We observe that the only admissible solutions occur at higher values of $\beta$ and have wave speeds $c$ increasing as functions of $\beta$, in line with the branches found in~\cite{Avitabile2023} for the case with $R=0$.
As $R$ increases across the three diagrams, the stable and unstable branches of the fast two-spike wave solution pinch together and produce two fold bifurcations.  For our chosen parameter values, admissible stable solutions exist on either side of this split, meaning we have an intermediate gap in the values of $\beta$ for which we can find this two-spike solution.  This contrasts with the results for $R=0$ in~\cite{Avitabile2017}, wherein no folds other than those at the extreme left of Figure \ref{fig:triptych-beta} were observed along wave branches.

\label{sec:Rbeta_bif}
\begin{figure}
    \centering
    \includegraphics[width=\textwidth]{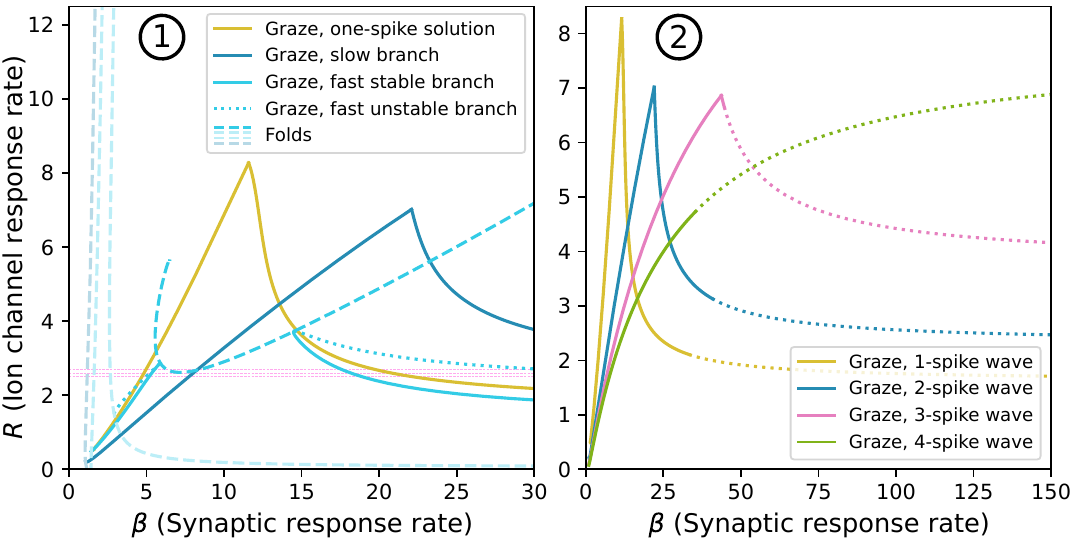}
    \caption{Bifurcation diagrams in $R$ and $\beta$. 1: Values of $R$ and $\beta$ at which graze and fold bifurcations occur for one- and two-spike wave solutions, akin to Figure \ref{fig:PAC-solutions-2spike}.  Pink dashed lines correspond to the values of $\beta$ producing the plots in Figure \ref{fig:triptych-beta}.  As shown in Figure \ref{fig:triptych-beta}, each branch of solutions folds at a low value of $\beta$; the graze is only tracked along the upper branch.  2: Grazing bifurcations exhibiting a double graze for 1-, 2- and 3-spike waves, and the grazing bifurcations for the analogous 4-spike wave.  The transition to dotted lines indicates the co-incidence of a transcritical bifurcation.}
    \label{fig:Rbeta_diagram}
\end{figure}

Figure \ref{fig:Rbeta_diagram}.1 tracks the bifurcations observed in Figure \ref{fig:triptych-beta} across the $(R, \beta)$ parameter plane following a similar convention as Figure \ref{fig:RD_diagram}.
As in Figure \ref{fig:RD_diagram}, we observe that the grazes of the fast two-spike wave closely approximate those of the one-spike wave until the former are interrupted by a fold bifurcation around the region defined by the one-spike double graze.  The double grazes of the one-spike and two-spike atomic solutions are more prominent in this view, defining intermediate regions of $\beta$ for which admissible solutions exist at higher values of $R$ than occurs at higher or lower values of $\beta$.
The double graze for the two-spike solution occurs at a higher value of $\beta$ than the double graze for the one-spike solution; in Figure \ref{fig:Rbeta_diagram}.2, we observe that this trend continues for the atomic three-spike wave, but the double graze for the atomic four-spike wave does not occur as $\beta$ increases before the timescales of $u$ and $s$ begin to separate and, furthermore, the solutions lose stability to a transcritical bifurcation.  Here, we remark that the previous results on stability of waves in the $R=0$ case were dominated by Hopf bifurcations~\cite{Avitabile2023}, representing oscillatory instabilities, rather than transcritical bifurcations with positive real roots.

\subsection{Simulations across distinct bifurcation scenarios under variation of the synaptic timescale}
\noindent We now repeat the numerical simulations as in Figure \ref{fig:diptych_R} except now varying $\beta$ instead of $R$.  Figures \ref{fig:triptych-beta}.1 and \ref{fig:triptych-beta}.2 show numerical simulations as $\beta$ is increased or decreased across the lower and upper folds present in the $R=2.7$ case, respectively.
We again see the slow transition characteristic of fold dynamics and delayed bifurcation phenomena under slow parameter variation.  However, we now see in Figure \ref{fig:triptych-beta}.1 that the initial bump is transient, with its highly synchronised structure breaking up into two secondary bumps.  In Figure \ref{fig:diptych_R}.2, the backwards-travelling component of the initial graze does not terminate at the bump edge, but instead continues as a one-spike wave.
When this one-spike wave loops around and meets the bump once more, instead of transitioning to a second bump as in Figure \ref{fig:diptych_R}.3, it instead reflects off the bump.  This reflection repeats multiple times, and the wave variably gains and loses spikes, eventually generating two additional bumps before terminating.  We additionally see small numbers of neurons locking into a bistable tonic firing state, in some cases persisting even as waves move across the neuron.

Reflection or deflection in laterally-inhibited spiking systems has been previously observed in one \cite{Avitabile2023} and two \cite{Gong2012} spatial dimensions, with waves and bumps maintaining a clear aura of nonfiring neurons due to their inhibitory influence.  In the latter model, the refractory period of the neuron was found to be a key parameter in enabling deflections, as a shorter refractory period allowed the firing pattern to re-enter its own wake.  Our model lacks an explicit refractory period, and we instead see that the post-firing voltage maxima provide a natural re-ignition point when lateral inhibition is cut off as the wave terminates.  However, analysis of what determines whether this leads to a bump or to a reflected wave is left to future work.

\section{Discussion}
\label{sec:discussion}
\noindent In this paper, we have investigated dynamics of a spiking neural field model comprising a synaptically-coupled network of integrate-and-fire neurons, extending results in~\cite{Laing2001,Avitabile2023}, by incorporating subthreshold oscillations, generated for example through adaptation ion currents, into the local dynamics.
We have analysed the dynamics of this system through a combination of bifurcation analysis and numerical simulation.
For the former, we have extended the stability results in~\cite{Avitabile2023} for travelling waves in spiking neural fields; for the latter, we have developed a novel, GPU-accelerated simulation algorithm that takes advantage of the form of the IF network to update the state of the network in an efficient and accurate manner.

We find that the slow adaptive ion channel increases wave speed through its interaction with the lateral inhibition modelled by the Mexican hat connectivity function \eqref{eq:mexican-hat}, while also accentuating the decrease in wave speed present in composite waves as spike count $m$ increases.  We also see that as $R$ increases, these waves gain stability through Hopf bifurcations while being rendered inadmissible through grazing bifurcations.

Interestingly, we find an intricate structure organising two-spike travelling wave solutions with different wave speeds.
These waves may be thought of as a composite of one spike waves that become `locked' together.
This locking behaviour is tightly linked to the natural frequencies of the individual neurons, and is similar to the behaviour observed in fluid dynamics models in which the tails of travelling convection cells become bound to each other~\cite{bergeonDynamicsFormationLocalized2010}.
This phenomenon may also play a role in discriminating situations in which travelling waves are reflected by bump attractors and those in which the waves are terminated as they reach the bump.
Although detailed analysis of this scenario is challenging since there is no general mathematical definition of what constitutes a bump in a spiking neural field model, progress could be made by considering specific types of bump solution, such as that demonstrated by recent analysis of the Lighthouse model~\cite{Coombes2025}.

Our IF network model is an extension of the LIF model considered in previous studies of travelling waves and bump attractors~\cite{Laing2001, Avitabile2023}.
Despite its additional complexity, the revised model is still a considerable simplification of more realistic models such as those following the Hodgkin--Huxley formalism, for which travelling wave solutions can be constructed after taking certain limits~\cite{muratovQuantitativeApproximationScheme2000}.
Finding travelling wave solutions in more realistic neuron models away from such limits is a difficult task, and often requires resorting to numerical approaches for both wave construction and stability.
One possible approach to overcome this barrier is to construct local dynamics using \textit{dynamic input conductances}~\cite{drionDynamicInputConductances2015}, which consolidate all neuronal currents into fast, slow, and ultraslow components, with timescales that potentially overlap.
Adopting such an intermediate complexity approach may allow for analytical progress in studying neural waves whilst providing a means to link the results back to biological mechanisms such as the relative expression of specific types of ion channel~\cite{brandoitFastReconstructionDegenerate2025a}.

Beyond the simplified single-neuron dynamics, we also take a simplified connectivity kernel $w$, combining the excitatory and inhibitory coupling into a single Mexican hat function.
Although this simplifies analysis, such an approach may miss important network dynamics.
For example, networks involving distinct excitatory and inhibitory populations have been shown to support slow one-spike waves in which inhibitory cells fire in advance their excitatory counterparts~\cite{Golomb2001, Golomb2002}.
Similarly, relaxation to bump attractors following perturbation has been shown to be delayed in coupled excitatory--inhibitory networks due to the repelling effect of the inhibitory population's bump on the excitatory population's bump~\cite{cihakDistinctExcitatoryInhibitory2022}.
Consideration of a separated excitatory--inhibitory network may be particularly relevant when considering local dynamics that represents individual spiking neurons so as to adhere to Dale's principle that neurons secrete only one primary type of neurotransmitter.

The bulk of our analysis in this work focuses on travelling waves with a small number of spikes, for which we see the impact of the natural frequency of the subthreshold oscillation on the `locking' behaviour in composite waves.
For travelling waves with higher numbers of spikes---essentially, travelling bursts of spikes---it is possible that the natural frequency, or even a resonant frequency, may play a different role in determining either the existence or stability of different types of wave.  Alternatively, it may simply modify the basins of attraction of such solutions.

In a similar vein, we have only considered travelling wave solutions in which every neuron in the domain fires the same number of times as the wave propagates.
Propagating waves in which not all neurons participate, or those in which neurons may fire a different number of times from each other, are commonly observed in biological neural networks.
These behaviours may be generated through stochasticity, heterogeneity, or anisotropy of the network.
One approach for addressing the first of these options could be to use Hawkes processes~\cite{locherbachSpikingNeuronsInteracting2017} to describe the local dynamics.
This approach essentially describes how the propensity of neuronal firing changes in response to synaptic inputs, treating the firing times themselves as random variables.
More complex spatiotemporal waves such as \textit{lurching waves} have been observed in networks with separated excitatory and inhibitory populations~\cite{golombContinuousLurchingTraveling1999}.

Given the consideration of oscillatory dynamics at the neuron scale, it is pertinent to ask how oscillations may operate across scales, particularly since it is known that the collective activity of neural populations generates rhythms in frequency bands that are associated with distinct functional states.
One prominent example of cross-scale resonances involves the generation of slow cortical waves, which are thought to play an important role in memory consolidation during sleep~\cite{miyamotoRolesCorticalSlow2017, neskeSlowOscillationCortical2016}.
These low-frequency waves, which propagate at speeds of around 2--7 ms$^{-1}$~\cite{massiminiSleepSlowOscillation2004}, are generated by local spiking activity that takes place at much higher frequency than the waves themselves.
The generation of such waves involve transitions of the local dynamics between a highly excitable and a less excitable state, referred to as \textit{up} and \textit{down} states respectively~\cite{neskeSlowOscillationCortical2016}.
Such solutions may be generated in our model through the incorporation of a further ultraslow timescales to mediate the transition between up and down states such that the local dynamics exhibits bursting behaviour \cite{Sciamanna2011}.
In general, extending our analysis to a model with three-dimensional local dynamics is a non-trivial task, since it introduces the possibility of generating more complex spatiotemporal patterns.
However, some progress may be possible by choosing piecewise linear models of bursting behaviour, such as that put forward in~\cite{Desroches2016}, which allow for construction of semi-explicit solutions in a similar vein to the approach taken in this manuscript.

Although our analysis has is restricted to a one-dimensional tissue, the analysis and numerical simulation algorithms can be extended to higher dimensions.
Simulations of our network in two-dimensions (not shown) demonstrate the existence of a variety of bumps and wave solutions, including plane waves and localised travelling structures resembling the gliders associated with the Game of Life, which have also been observed in cellular automata representations of neural networks~\cite{Gong2012}.
Analysis of such solutions will involve considering different classes of perturbation to wavefronts that incorporate the different instabilities that can occur in higher dimensions~\cite{bressloffLaminarNeuralField2015, visserStandingTravellingWaves2017}.
Exploration of such models will be important to capture any relevant information regarding local geometry of \textit{in vivo} neural tissue, or \textit{in vitro} organoid models, as has been done for the macroscopic Amari-style neural fields over arbitrary surfaces~\cite{shawRadialBasisFunction2025}.

\subsection*{Acknowledgements}

This work was supported by the Engineering and Physical Sciences Research Council.
For the purpose of open access, we have applied a Creative Commons Attribution (CC BY) license to any Author Accepted Manuscript version arising from this submission. 

\printbibliography
\clearpage

\appendix

\section{Simplifying the single-neuron equations}
\label{app:vu-noinput}

\noindent Equations \ref{eq:v-no-input} and \ref{eq:u-no-input} for the single-neuron values of $v$ and $u$ between events are somewhat unwieldy.  We may simplify these equations by following the general addition formula for trigonometric functions (with $X \neq 0$),
\begin{equation}
    X \cos(t) + Y \sin(t) = 
    \text{sgn}(X)\sqrt{X^2 + Y^2} \cos\left(t - \arctan \left( \frac{Y}{X} \right) \right),
\end{equation}
whereby for 
\begin{align}
    X_1 &= v_k(0) - s_k(0) \frac{2p - (\beta + 1)}{(p - \beta)^2 - q^2} - I \frac{2p - 1}{p^2 - q^2},\\
    Y_1 &= \frac{1}{|q|} \bigg(
    s_k(0) \frac{p^2 + q^2 -p(\beta + 1) + \beta}{(p - \beta)^2 - q^2}
     + I \frac{p^2 + q^2 - p}{p^2 - q^2}
     + v_k(0)(1 - p) + u_k(0) \bigg),\\
    X_2 &= u_k(0) - R \bigg(
    \frac{s_k(0)}{(p - \beta)^2 - q^2} + \frac{I}{p^2 - q^2}\bigg),\\
    Y_2 &= \frac{1}{|q|}\bigg(R \bigg(
    s_k(0) \frac{p + \beta}{(p - \beta)^2 - q^2} + I \frac{p}{p^2 - q^2} - v_k(0)\bigg)
     + (p - 1) u_k(0) \bigg),
\end{align}
we can further derive the following parameters
\begin{align}
    K_1 &= \sqrt{X_1^2 + Y_1^2}, & 
    K_4 &= \sqrt{X_2^2 + Y_2^2},\\
    \theta_1 &= \arctan \left( \frac{Y_1}{X_1} \right), &
    \theta_2 &= \arctan \left( \frac{Y_2}{X_2} \right),\\
    K_2 &= s_0 \frac{2p - (\beta + 1)}{(p - \beta)^2 - q^2}, &
    K_5 &= \frac{Rs_0}{(p - \beta)^2 - q^2},\\
    K_3 &= I \frac{2p - 1}{p^2 - q^2}, &
    K_6 &= \frac{RI}{p^2 - q^2},
\end{align}
to write
\begin{align}
    v_k(t) &= K_1 e^{-pt} \cos(|q|t + \theta_1) + K_2 e^{-\beta t} + K_3,\\
    u_k(t) &= K_4 e^{-pt} \cos(|q|t + \theta_2) + K_5 e^{-\beta t} + K_6.
\end{align}

\section{Numerical simulation}
\label{app:numerics}
\noindent For the purposes of simulating our discrete model \eqref{eq:v_discrete}--\eqref{eq:s_discrete} and using it to verify analysis of our continuum model \eqref{eq:v_continuum}--\eqref{eq:s_continuum}, we have constructed an event-based simulator for our network.  From given initial conditions, our neurons evolve under equations \eqref{eq:v_discrete}--\eqref{eq:s_discrete} with $f_n^\text{in} = 0$ and no voltage resets until a neuron fires somewhere within the network.  This means their behaviour is described by the analytical solutions \eqref{eq:v-no-input}--\eqref{eq:u-no-input} at times when there are no neurons firing.  As such, we can solve \eqref{eq:v-no-input} to find the next time we attain $v_n = v_\text{th}$ for each neuron $n$ then update the variables across the entire network to that firing time, reset the firing neuron, and propagate the synaptic signal.  By repeating this process we can evolve the state of the network without discretising the governing differential equations, enabling us to find our network firing times to machine precision.

As discussed in Appendix \ref{app:vu-noinput}, our solved equation \eqref{eq:v-no-input} for each $v$ (hereon omitting the neuron indexing) between firing events and with no inputs (i.e. $f_\text{in} \equiv 0)$ has the form
\begin{equation}
    v(t) = K_1 e^{-pt} \cos(|q|t + \theta_1) + K_2 e^{-\beta t} + K_3,
    \label{eq:v_simple}
\end{equation}
where the real constants $K_1$, $K_2$, $K_3$, $\theta_1$ are known, and $K_1$, $K_2$, $\theta_1$ depend upon the initial conditions of $v$, $u$ and $s$; without loss of generality, we take our initial condition to occur at $t_0 = 0$.  As long as we can find the least $t^* \geq 0$ such that $v(t^*) = v_\text{th}$, we can find the next firing time of the neuron, and thus evolve our network in the event-based manner described above.  However, our equation for $v(t)$ does not admit an analytical inverse, and its oscillatory nature hinders us in establishing both existence of the root and basins of root-finding convergence, which prevents us from applying common root-finding algorithms that expect such information to be known \cite{Chapra2005, Gerald1999, Froberg1985, Conte1980, Hornbeck1975}.  As such, we have designed a custom root-finding algorithm to meet our needs, similar to the more involved approach of \textcite{Makino2003}.

The first stage of our algorithm identifies a finite interval outside of which $t^*$ cannot exist.  The second stage of our algorithm conducts a forward search through said interval that will either converge to the root $t^*$ or present a signal in a finite number of steps to indicate that no root exists.

\subsection{Constructing a solution interval}
We have set our initial condition at $t = 0$, and $v(0) < v_\text{th}$ by design, so we know trivially that $t^* > 0$.  From \eqref{eq:v_simple}, we see that as $t \to \infty$, $v(t) \to K_3$.  For the general case $K_3 \neq v_\text{th}$, we can bound this convergence by
\begin{equation}
    |v(t) - K_3| \leq |K_1|e^{-pt} + |K_2|e^{-\beta t} \leq \big(|K_1| + |K_2|\big)e^{-\min\{p, \beta\} t},
\end{equation}
making use of $p, \beta, t \geq 0$.  This allows us to find a bound $T > 0$ such that $\forall t > T$ we have $|v(t) - K_3| < |v_\text{th} - K_3|$, given by
\begin{equation}
    T = \frac{1}{\min\{p, \beta\}} \log \bigg( \frac{|v_\text{th} - K_3|}{|K_1| + |K_2|} \bigg).
\end{equation}
There can be no solution $v(t^*) = v_\text{th}$ for $t^* > T$, so either $t^* \in [0, T]$ or $t^*$ does not exist.  We do not consider the case $K_3 = v_\text{th}$ as it represents a measure-zero subset of all possible parameters.  

\subsection{Forward-search root-finding algorithm}
\label{numerics:rootfinder}
We have a finite interval $[0, T]$, but we do not know whether a root $t^*$ exists in this interval. To approach this situation, we seek to construct a sequence $\{t_n\}_{n \in \mathbb{Z}}$ such that $t_n \to t^*$ if $t^*$ exists, and that if $t^*$ does not exist we find a $t_n > T$ or demonstrate that $v(t)$ is nonincreasing for $t > t_n$ to clearly signal the nonexistence of a solution.  To achieve this, we make use of a modified Newton-Raphson method with iterative formula
\begin{equation}
    t_{n+1} = t_{n} + \frac{v_\text{th} - v(t_n)}{m_n}, \quad n = 0,1,2, \dots
    \label{eq:root_find}
\end{equation}
where $t_0 = 0$, $v(t_0) < v_\text{th}$ and $m_n$ is given by
\begin{equation}
    m_n = \min \bigg\{M_n,\ \frac{1}{2}v'(t_n) + \frac{1}{2}\sqrt{v'(t_n)^2 + 4 \big(v_\text{th} - v(t_n) \big)M'_n } \bigg\},
\end{equation}
for
\begin{equation}
\begin{split}
    M_n =&\ (p + |q|)|K_1| e^{-pt_n}  + \beta \max\left\{- K_2 e^{-\beta t_n}, - K_2 e^{-\beta T}\right\},\\
    &\geq \sup\big\{ v'(t) : t \in [t_n, T] \big\}, \label{eq:Mn}
\end{split}
\end{equation}
\begin{equation}
\begin{split}
    M'_n =&\ \max\left\{(p + |q|)^2 |K_1| e^{-pt_n} + \beta^2 
    \max\left\{K_2 e^{-\beta t_n}, K_2 e^{-\beta T}\right\}, 0\right\},\\
    &\geq \sup\big\{ v''(t) : t \in [t_n, T] \big\}. \label{eq:M'n}
\end{split}
\end{equation}
We halt when $t_{n+1} - t_n < \epsilon_\text{converge}$ for some small constant $\epsilon_\text{converge}$ as a proxy for convergence, as standard.  Additionally, halt if $t_{n+1} > T$ to signify that we have exited the interval.  We also halt immediately if $M_n \leq 0$, as then $v(t)$ is decreasing on the remainder of the interval and no root will be found.  We shall prove below that $v(t_n) < v_\text{th}$ and $M_n \geq \sup\{v'(t) : t \in [t_n, T]\}$, so if $M_n \leq 0$ then $v(t)$ is nonincreasing on the remainder of the interval, and as such there can be no solution $v(t^*) = v_\text{th}$.  We set $M'_n \geq 0$ to ensure that our second option for the value of $m_n$ always evaluates as real.

This algorithm functions by setting $m_n$ as an upper bound on $v'(t)$ for $t \in [t_n, t_{n+1}]$ so that if $t^*$ exists then each iteration $t_{n+1}$ is an underestimate of $t^*$.  This prevents the algorithm from overshooting any roots, and leads to its convergence to $t^*$ as proven below.  Our choosing the lesser of two bounds for the value of $m_n$ is motivated by improving the rate of convergence; near to the root, $v_\text{th} - v(t_n)$ is small, so the $M'_n$-derived option converges to the original Newton--Raphson method, thus recapturing its quadratic convergence.  Further from the root, the bound $M_n$ is smaller and thus gives larger step sizes $t_{n+1} - t_n$.

To prove that this method has the properties we require, Lemmas \ref{lemma:Mn_bound} and \ref{lemma:M'n_bound} guarantee that $m_n$ overestimates $v'(t)$ on the interval $[t_n, t_{n+1}]$, while Lemma \ref{lemma:bounded} ensures $\{t_n\}$ is bounded above by $t^*$ if it exists.  This then enables Proposition \ref{prop:converges} to show that $t_n \to t^*$ if $t^*$ exists, giving us our solution, and Proposition \ref{prop:diverges} to show that $t_n \to \infty$ if $t^*$ does not exist, allowing us to detect our sequence $\{t_n\}$ exiting the interval $[0, T]$ to show no solution exists.

\begin{lemma}
\label{lemma:bounded}
    Suppose $\exists t^*$ such that $v(t^*) = v_\text{th}$.  Let $t_n < t^*$, $v(t_n) < v_\text{th}$, $m_n > 0$.  If $m_n \geq v'(t)\ \forall t \in [t_n, t_{n+1}]$, then $t_{n+1} < t^*$.
\end{lemma}
\begin{proof}
    For $t \in [t_n, t_{n+1}]$ we have
    \begin{equation}
        v(t) = v(t_n) + \int_{t_n}^t v'(\tau) d\tau \leq v(t_n) + (t - t_n)m_n \leq v(t_n) + (t_{n+1} - t_n)m_n = v_\text{th}. 
    \end{equation} 
    Equality can only be achieved if both $t = t_{n+1}$ and $m_n = v'(t)$ $\forall t \in [t_n, t_{n+1}]$.  As $v$ is not a linear function, $v(t) < v_\text{th}$ $\forall t \in [t_n, t_{n+1}]$ and so $t_{n+1} < t^*$.
\end{proof}

\begin{lemma}
\label{lemma:Mn_bound}
    For a given iteration $n$, $M_n > v'(t)$ $\forall t \in [t_n, T]$.
\end{lemma}
\begin{proof}
    From Equation \eqref{eq:v_simple} we have
    \begin{equation}
        v'(t) = \big(-p\cos(|q|t + \theta) - |q|\sin(|q|t + \theta)\big)K_1e^{-pt} - \beta K_2 e^{-\beta t}.
    \end{equation} 
    As $p, \beta > 0$ and $\cos(|q|t + \theta) \neq \sin(|q|t + \theta)$ we have
    \begin{equation} 
        v'(t) < (p + |q|)|K_1|e^{-pt} - \beta K_2 e^{-\beta t},
    \end{equation} 
    and as $t_n < t < T$ we have
    \begin{equation} 
        (p + |q|)|K_1|e^{-pt} - \beta K_2 e^{-\beta t} \leq (p + |q|)|K_1|e^{-pt_n} + \beta \max\left\{-K_2 e^{-\beta t_n}, -K_2 e^{-\beta T}\right\} = M_n,
    \end{equation} 
    where the choice between $e^{-\beta t_n}$ and $e^{-\beta T}$ depends upon the sign of $K_2$.  Therefore $M_n > v'(t)$.
\end{proof}

\begin{lemma}
\label{lemma:M'n_bound}
    For a given iteration $n$, let $m_n = \frac{1}{2}v'(t_n) + \frac{1}{2}\sqrt{v'(t_n)^2 + 4 \big(v_\text{th} - v(t_n) \big) M'_n}$.
    Then $m_n \geq v'(t)$ $\forall t \in [t_n, t_{n+1}]$.
\end{lemma}
\begin{proof}
    Following Lemma \ref{lemma:Mn_bound}, from Equation \eqref{eq:v_simple} we have 
    \begin{equation}
        v''(t) = (-p^2 - |q|^2)\cos(|q|t + \theta)K_1e^{-pt} + \beta^2 K_2 e^{-\beta t},
    \end{equation} 
    giving us
    \begin{equation} 
        v''(t) \leq (p^2 + |q|^2)|K_1|e^{-pt_n} + \beta^2 \max\left\{K_2 e^{-\beta t_n}, K_2 e^{-\beta T} \right\} \leq M'_n, 
    \end{equation} 
    This leads us to
    \begin{equation}
    \begin{split}
        v'(t) =&\ v'(t_n) + \int_{t_n}^t v''(\tau)d\tau \leq v'(t_n) + (t - t_n)M'_n \\
        &\leq v'(t_n) + (t_{n+1} - t_n) M'_n = v'(t_n) + M'_n\frac{v_\text{th} - v(t_n)}{m_n}.
    \end{split}
    \end{equation}
    We can therefore prove $v'(t) \leq m_n$ by showing
    \begin{equation}
        m_n \geq v'(t_n) + M'_n\frac{v_\text{th} - v(t_n)}{m_n}, 
    \end{equation} 
    which is a quadratic inequality solved by
    \begin{equation}
        m_n \geq \frac{1}{2}v'(t_n) + \frac{1}{2}\sqrt{v'(t_n)^2 + 4 \big(v_\text{th} - v(t_n) \big) M'_n}. 
    \end{equation} 
\end{proof}

\begin{proposition}
\label{prop:converges}
    If $t^*$ exists, then $t_n \to t^*$ as $n \to \infty$.
\end{proposition}
\begin{proof}
    Assume $t^*$ exists.
    We know that $t^* \in [0, T]$.  
    We have $t_0 = 0$ and $v(0) < v_\text{th}$, so by induction on Lemmas \ref{lemma:bounded}--\ref{lemma:M'n_bound} we have $t_n < t^*$ $\forall n$.  
    We have $m_n > 0$, and $v(t_n) < v_\text{th}$, so from the formula \eqref{eq:root_find} we have $t_n < t_{n+1}$.
    Therefore, $\{t_n\}_{n \in \mathbb{N}}$ is a monotonic increasing sequence bounded from above, and so it must converge to some limit $\overline{t}$.  Then as $n \to \infty$, we have
    \begin{equation} 
        (t_{n+1} - t_n)m_n \to v_\text{th} - v(\overline{t}), 
    \end{equation} 
    and 
    \begin{equation}
        (t_{n+1} - t_n) \to 0. 
    \end{equation} 
    As $\lim_{n \to \infty} (m_n)$ is finite, the only way for the first limit to hold is if $v_\text{th}$ = $v(\overline{t})$, and so by definition $\overline{t} = t^*$.  
\end{proof}

\begin{proposition}
\label{prop:diverges}
    If $t^*$ does not exist, then $t_n \to \infty$ as $n \to \infty$.
\end{proposition}
\begin{proof}
    We have $v(t_0) < v_\text{th}$, so if $t^*$ does not exist then $v(t) < v_\text{th}$ $\forall t > t_0$, and as $m_n > 0$ by Equation \eqref{eq:root_find} we have that $\{t_n\}_{n \in \mathbb{N}}$ is a strictly increasing sequence.

    As $\{t_n\}$ is strictly increasing, either it converges to some limit $t_n \to \overline{t}$ as $n \to \infty$, or it diverges $t_n \to \infty$.  If $t_n \to \overline{t}$, then by the same construction as in Proposition \ref{prop:converges}, $v(\overline{t}) = v_\text{th}$ and $\overline{t} = t^*$, raising a contradiction.  Therefore $t_n \to \infty$ as $n \to \infty$.
\end{proof}

\noindent A pseudocode description of this root-finding algorithm is presented as Algorithm \ref{alg:rootfinding}.

\begin{algorithm}
\caption{Root-finding for firing times.  For a neuron $k \leq N$, returns $t^*_k$ as the next firing time of the neuron (measured from the present time $t_\text{global}$), or returns nothing if the neuron will not fire.}
\label{alg:rootfinding}
\begin{algorithmic}
    \State $v_0 := v_k(t_\text{global}),\ u_0 := u_k(t_\text{global}),\ s_0 := s_k(t_\text{global})$
    \State Calculate $K_1,\ K_2,\ K_3$
    \State $T_k = -\frac{1}{\min\{p, \beta\}} \log \left( \frac{|v_\text{th} - K_3|}{|K_1| + |K_2|} \right)$ \Comment{Construction of solution interval}
    \State $t_0 := t_\text{glob}$ 
    \For{$n \in \mathbb{N}$}
        \State $M_n := (p + |q|)|K_1| e^{-pt_n}  + \beta \min\left\{-K_2 e^{-\beta t_n}, -K_2 e^{-\beta T}\right\}$
        \If{$M_n \leq 0$}
            \State Return $\emptyset$
        \EndIf
        \State $M'_n := \max\left\{(p^2 + |q|)^2 |K_1| e^{-pt_n} + \beta^2 
        \max\left\{ K_2 e^{-\beta t_n}, K_2 e^{-\beta T}\right\}, 0\right\}$
        \State $m_n := \min \bigg\{M_n,\ \frac{1}{2}v'(t_\text{glob} + t_n) + \frac{1}{2} \sqrt{v'(t_\text{glob} + t_n)^2 + 4\big(v_\text{th} - v(t_\text{glob} + t_n) \big) M'_n} \bigg\}$
        \State $t_{n+1} := t_{n} + \frac{v_\text{th} - v(t_\text{glob} + t_n)}{m_n}$
        \If{$t_{n+1} > T_k$}
            \State Return $\emptyset$
        \ElsIf{$t_{n+1} - t_n \leq \epsilon_{converge}$}
            \State Return $t^*_k:= t_{n+1}$
        \EndIf
    \EndFor
\end{algorithmic}
\end{algorithm}

\subsection{Evolving the network}
The above algorithm tells us either that a given neuron with no further inputs will not fire, or how far into the future it will fire.  By applying this to each neuron in a network, we either find that none of the neurons will fire (and so the network returns to a rest state, and we can stop the simulation), or we find the earliest time that a neuron will fire.  In the latter case we can update each neuron's variables $v$, $u$, $s$ to that firing time, reset the firing neuron's voltage to $v_\text{r}$, and respond to the firing by updating the value of $s$ for each other neuron depending upon their connection to the firing neuron as determined by $w$.  By repeating this process, we can evolve the network directly from one firing time to the next.  Taking $t_\text{global}$ as the global timepoint for our model, and for given initial conditions, our algorithm to evolve the network is expressed in pseudocode as Algorithm 2.

In certain cases, such as a planar wave on a domain with dimension 2 or higher and with its wavefront parallel to the neuron lattice, we expect to see many neurons firing simultaneously.  For such situations we can include an additional step in which we consider all neurons with a firing time marginally slower than the fastest firing time to instead fire simultaneously at that fastest firing time, thus allowing us to process all such firing events simultaneously.  Otherwise, finite-precision errors and marginal differences in root-finding results can create small differences in the expected-to-be-equal firing times, forcing the algorithm to loop multiple times as it makes very small steps forwards in model-time.  However, this can potentially suppress slow-growing instabilities or interfere with solutions that genuinely expect slight differences in firing times, so some care should be taken in its use.

\begin{algorithm}
\caption{Global firing event determination and network evolution.  Takes $M$ as the minimum number of firing events to find, and initial conditions $v_k(0)$, $s_k(0)$ for each $k \in \mathbb{N}_N$.  Returns $\mathcal{R}$, a set of all detected firing times as pairs of times and neuron indices firing at that time.}
\label{alg:firingtime}
\begin{algorithmic}
    \State $t_\text{global} := 0$
    \State $\mathcal{R} = \emptyset$ \Comment{Set of recorded spikes}
    \While{$\#\mathcal{R} < M$:}
        \State $\mathcal{F} := \emptyset$, $\mathcal{K} := \emptyset$
        \For{$k \in \mathbb{N}_N$}
            \State Call root-finding algorithm (Algorithm 1) on neuron $k$
        \EndFor
        \If{$\mathcal{F} \neq \emptyset$}
            \State $t = \min\mathcal{F}$
            \For{$k \in \mathbb{N}_N$} \Comment{Optional optimisation for synchronous firing}
                \If{$t < t^*_k \leq (1 + \epsilon_\text{firing})$}
                    \State Include $k$ in $\mathcal{K}$
                \EndIf
            \EndFor
            \For{$k \in \mathbb{N}_N$} \Comment{Updating the network to the next firing time}
                \State Calculate $v_k(T + t)$, $u_k(T + t)$, $s_k(T + t)$ without considering firing
                \If{$k \in \mathcal{K}$}
                    \State $v_k(T + t) := v_\text{th}$
                \EndIf
                \For{$j \in \mathcal{K}, j \neq k$}
                    \State $s_k(T + t) := s_k(T + t) + \Delta \beta w(||\mathbf{x}_j - \mathbf{x}_k||)$
                \EndFor
            \EndFor
            \State $t_\text{global} := t_\text{global} + t$
            \State Record $(t_\text{global}, \mathcal{K})$ in $\mathcal{R}$
        \EndIf
    \EndWhile   
\end{algorithmic}
\end{algorithm}

Beyond explaining our method, we also wish to highlight the ``embarrassingly" parallel nature of this problem.  The only step which requires drawing information from multiple neurons is the comparison to find the earliest firing time.  Otherwise, each neuron can have its own firing time calculated in isolation, and once it receives a timepoint and a list of firing neurons its new variable values can also be calculated in isolation.  To take advantage of this potential for parallelisation, we have implemented the code simulating the network as a parallel program to be run on NVidia graphics cards using the CUDA interface, written in Python \cite{Numba} and available on github at \url{https://github.com/henrydjkerr/Timestep-free-evolver}.  The code was written with modular implementation of most features of the model (such as governing equations, root-finders, variable count and initialisation) to allow variations on the model to be explored easily.

\end{document}